%% file: ys-202604-no-phonon-bec.tex
\documentclass[11pt,a4paper]{article}

\usepackage{fontspec}
\usepackage[titletoc,title,header]{appendix}

\usepackage{amsmath}
\usepackage{amssymb}
\usepackage{amsthm}
\usepackage{mathtools}
\usepackage{stmaryrd}  
\usepackage{xparse}    

\usepackage[margin=1in]{geometry}

\usepackage{graphicx}
\usepackage{xcolor}

\usepackage{hyperref}
\hypersetup{
  colorlinks=true,
  linkcolor=blue,
  citecolor=blue,
  urlcolor=blue,
  pdfauthor={Yoshitsugu Sekine},
  pdftitle={No-Go Theorem for Quasiparticle BEC}
}

\usepackage[]{natbib}
\theoremstyle{plain}
\newtheorem{thm}{Theorem}[section]

\newtheorem{prop}[thm]{Proposition}
\newtheorem{cor}[thm]{Corollary}

\newtheorem{fact}[thm]{Fact}

\theoremstyle{definition}
\newtheorem{defn}[thm]{Definition}

\theoremstyle{remark}
\newtheorem{rem}[thm]{Remark}

\input{mycommands.tex}

\title{No-Go Theorem for Quasiparticle BEC}

\author{%
Yoshitsugu Sekine\\{\small\texttt{4429sekine@gmail.com}}%
}

\date{\today}

\begin{document}

\maketitle

\begin{abstract}
We discuss a no-go theorem for Bose-Einstein condensation (BEC) of quasiparticles (phonons) from the viewpoint of operator algebras, using the van Hove model. The \(\beta\)-KMS states of the van Hove model satisfy the self-consistency condition of \cite{VIYukalov001}. However, the self-consistency condition is a constraint concerning the definition of the field, and is insufficient to establish the no-go theorem for BEC. In this paper, we prove the no-go theorem for BEC via two routes. First, imposing time cluster properties on the \(\beta\)-KMS states precludes BEC. Second, under nonlinear dispersion with \(s > 2\), the treatment of infrared divergences automatically reduces the algebra of physical observables, and BEC is mathematically excluded on the reduced algebra. In particular, the latter property admits an interpretation in terms of the ideal theory of the resolvent algebra.

\end{abstract}

\setcounter{tocdepth}{3}
\tableofcontents

\section{Introduction}\label{expedition0012078}

In \cite{VIYukalov001}, a no-go theorem for BEC of quasiparticles is discussed from the viewpoint of Hamiltonian design. That is, quasiparticles should be designed so that BEC is forbidden in their equilibrium states, and the no-go theorem for BEC of quasiparticles in equilibrium holds by definition. In particular, the self-consistency condition can be regarded as a condition that simultaneously determines the Hamiltonian and the associated equilibrium state. Here, we consider the converse problem: while there exist many widely accepted models involving quasiparticles, few of them have been rigorously verified to satisfy the no-go theorem for BEC. We therefore fix a concrete Hamiltonian, examine the existence of its equilibrium states, investigate whether BEC occurs in those equilibrium states, and discuss what conditions preclude BEC. We do not discuss BEC in non-equilibrium states. Furthermore, from the perspective of focusing on a model with the linear dispersion relation \(\omega(k)
= \abs{k}\) and a first-order perturbation of the field added to the free Hamiltonian, we refer to the quasiparticles specifically as phonons in the following discussion.

We first summarize the physical motivation of \cite{VIYukalov001}. The essential factor determining whether BEC can occur is the conservation of particle number. For conserved particles, lowering the temperature causes particles to accumulate in low-energy states, eventually leading to macroscopic occupation of the ground state. On the other hand, quasiparticles such as phonons, magnons, and bogolons are excitations generated by supplying energy from outside, and their density itself decreases with lowering temperature, ultimately vanishing. These quasiparticles have no independent particle number conservation law, and their chemical potential is fixed at \(0\), so there is no statistical-mechanical driving mechanism to push them into the ground state. Phenomena that often appear to be condensation of quasiparticles actually correspond to redefinitions of the ground state or the order parameter. Therefore, from any perspective---the properties of distribution functions, the density behavior in the low-temperature limit, or the conditions defining the field---BEC of non-conserved quasiparticles in equilibrium is excluded in principle.

We now explain the setting of this paper. We take the standpoint of verifying the claims of \cite{VIYukalov001} using operator algebras. Specifically, we fix the van Hove model \cite{AsaoArai26} and mathematically examine whether its equilibrium states actually exhibit BEC. The van Hove model is mathematically identical to the behavior of the phonon field in the Hubbard--phonon interaction system discussed in \cite{YoshitsuguSekine001,YoshitsuguSekine002}, and is also a concrete example of the generalized spin-boson model \cite{AraiHirokawa1,AraiHirokawa2}.

The main mathematical results are as follows. First, the \(\sminvtemperature\)-KMS states of the van Hove model satisfy the self-consistency condition of \cite{VIYukalov001}. However, the self-consistency condition is a constraint concerning the definition of the field, and this alone is insufficient to establish the no-go theorem for BEC. In this paper, we make this insufficiency explicit and prove the no-go theorem for BEC via two routes.

First, BEC in the van Hove model is described in the same form as BEC in the free Bose gas. This is controlled by the sesquilinear form \(\opform{q}_0\), and the results for the free Bose gas obtained in \cite{AsaoArai28,YoshitsuguSekine004} can be applied. Imposing time cluster properties---corresponding to the mildness of the equilibrium state---on the \(\sminvtemperature\)-KMS states yields the vanishing of BEC. The mildness of the equilibrium state for phonons in \cite{VIYukalov001} is a natural mathematical formulation of this time cluster property. It is also equivalent to the space cluster property, which corresponds to spatial mildness.

Second, when infrared divergences occur, under the higher-order nonlinear dispersion \(\omega_s(k)
= \abs{k}^s\) (\(s > 2\)), the treatment of infrared divergences automatically reduces the algebra of physical observables, and BEC is mathematically excluded on the reduced algebra. If the higher-order dispersion relation is renormalized into the background field in order to suppress the lattice collapse caused by infrared divergences, the reduction of the algebra of physical observables corresponds to this physical picture. In the free Bose gas or the van Hove model, the element responsible for generating BEC is carried by the sesquilinear form \(\opform{q}_0\), which generates the center of the von Neumann algebra obtained as the weak closure of the representation of the algebra of physical observables. According to the guiding principle presented in \cite{YoshitsuguSekine005}, the center describes quantities corresponding to classical components, and this is consistent with the interpretation that it represents a background field that should be removed.

Although \cite{YoshitsuguSekine005} emphasizes the formulation using the resolvent algebra rather than the conventional Weyl algebra, resolvents can be reduced to Weyl operators via the Laplace transform. Since computations with Weyl operators are more concise, we first establish the results on the Weyl algebra and then transplant them to the resolvent algebra. Preliminary computations on the Weyl algebra are carried out concretely on Fock space, and the representation-independent arguments proceed by axiomatically accepting these computational results.

The discussion of the van Hove model in this paper carries over directly to the Hubbard--phonon interaction system \cite{YoshitsuguSekine001,YoshitsuguSekine002}. Extensions to the spin-boson model \cite{FannesNachtergaeleVerbeure1,LorincziHiroshimaBetz3}, the Nelson model \cite{AsaoArai26,LorincziHiroshimaBetz3}, and more generally to many-body electron or electron-field--phonon interaction systems are left for future work. In this paper, we restrict ourselves to the behavior for the linear dispersion in three dimensions, emphasizing the physical setting. Appropriate generalizations and mathematical properties of the van Hove model will be discussed in a forthcoming paper. Although we employ techniques based on operator algebras \cite{BratteliRobinson1,BratteliRobinson2}, arguments based on functional integrals can also be constructed using the techniques of \cite{YoshitsuguSekine004}.

\section{Main Results}\label{main-results}

\subsection{Definitions}\label{definitions}

We set the basic complex Hilbert space and its real subspace as \[\sphilb{H}
=
\fun{\lp^{2}}{\fldreal^{d}},
\quad
\sphilb{H}_{\txtreal}
=
\fun{\lp_{\txtreal}^{2}}{\fldreal^{d}}
=
\fun{\lp^{2}}{\fldreal^{d};\fldreal}\]and assume \(d
= 3\) unless otherwise stated. The symplectic form is defined by \(\sigma(f,g)
= \opimag \bkt{f}{g}_{\sphilb{H}}\). As \(X\) for the real symplectic space \((X,\sigma)\), we choose the realification of the complex Hilbert space \(\sphilb{H}\), and the inner product on this real Hilbert space is given by \(\opreal \bkt{f}{g}\).

For a positive real number \(s
> 0\), the one-particle Hamiltonian (dispersion relation) defined on the wave-number space \(\greuctr{k}{d}\) is \(\omega(k)
= \omega_s(k)
= \abs{k}^{s}\). For a non-positive chemical potential \(\smchemicalpotential
\leq 0\), the non-negative self-adjoint operator is \(K_{\sminvtemperature,\smchemicalpotential}
=
\coth \frac{\sminvtemperature \rbk{\omega - \smchemicalpotential}}{2}\). For the associated non-degenerate non-negative symmetric sesquilinear form \(\opform{q}_{\txtnonzero,\smchemicalpotential}\), the associated inner product space and its completion are \[\sphilb{D}_{\sminvtemperature,\smchemicalpotential}
=
\pairbk{\fun{\opformdomain}{\opform{q}_{\txtnonzero,\smchemicalpotential}},\opform{q}_{\txtnonzero,\smchemicalpotential}},
\quad
\sphilb{H}_{\sminvtemperature,\smchemicalpotential}
=
\gtclos{\sphilb{D}_{\sminvtemperature,\smchemicalpotential}}^{\opform{q}_{\txtnonzero,\smchemicalpotential}}.\] We assume \(d
= 3\) and \(s
\geq 1\) in order to restrict to physically relevant situations.

\begin{rem}
In contrast to $s
= 1$, which is the original phonon dispersion relation, $s
> 1$ is regarded as a nonlinear term that suppresses excess phonons to prevent blow-up due to infrared and ultraviolet divergences,
and is introduced in order to consider the renormalization of this contribution into the background field.
As we shall see below,
the essential point is that the treatment strategy changes significantly between the original linear dispersion of acoustic phonons and nonlinear dispersions.
The same argument applies to other appropriate dispersion relations.
\end{rem}

In general, the bosonic Fock space over a Hilbert space \(\sphilb{H}\) is defined by \(\fun{\spfock_{\txtbsn}}{\sphilb{H}}
=
\bigoplus_{n=0}^{\infty}
\bigotimes_{\txtsym}^{n}
\sphilb{H}\), and for any \(f
\in \sphilb{H}\), the creation and annihilation operators on the bosonic Fock space are denoted by \(\opfockcran_{\txtfock}(f)\). The Segal field operator is defined by \[\opfocksegal_{\txtfock}(f)
=
\frac{1}{\sqrt{2}}
\rbk{\opfockcr_{\txtfock}(f) + \opfockan_{\txtfock}(f)}\]and the Weyl operator is defined by \(\opfockweyl_{\txtfock}(f)
=
\napiernum^{\imunit \opfocksegal_{\txtfock}(f)}\). The self-adjoint operator \(\physham_{\txtbsn,\txtfr}
= \fun{\opfocksndqntdiff_{\txtbsn}}{\omega}\) on Fock space is the second quantization operator determined by the dispersion relation, and its ground state is the Fock vacuum. The Hamiltonian of the van Hove model is defined as a perturbation of the free-field Hamiltonian by the Segal field operator with the source function \(\varrho\) as its variable; specifically, \[\physham_{\txtvanhowe}
=
\physham_{\txtbsn,\txtfr}
+\fun{\opfocksegal_{\txtfock}}{\frac{\varrho}{\sqrt{\omega}}}.\] By the Kato--Rellich theorem, this is a self-adjoint operator bounded from below. The infrared and ultraviolet cutoffs are imposed on the source \(\varrho\), and they will be specified in detail later.

\subsection{Quantities related to Bose-Einstein condensation}\label{quantities-related-to-bose-einstein-condensation}

Let \(\smnumberdensity_0(\sminvtemperature)\) denote the condensate density at inverse temperature \(\sminvtemperature
> 0\). We define the non-closed non-negative symmetric bilinear form corresponding to the condensate component as \[\opform{q}_{0}(f)
=
2 (2 \pi)^d \smnumberdensity_0(\sminvtemperature)
\abs{\faftr{f}(0)}^{2},
\quad
\opformdomain(\opform{q}_{0})
=
\fun{\lp^{1}}{\fldreal^{d}}
\cap
\fun{\lp^{2}}{\fldreal^{d}}\]and define the subspace \(\sphilb{D}_{0,\sminvtemperature,\smchemicalpotential}
=
\opformdomain(\opform{q}_0)
\cap
\sphilb{H}_{\sminvtemperature,\smchemicalpotential}\). When the value of the chemical potential \(\smchemicalpotential\) has no special significance, or when \(\smchemicalpotential
= 0\), we omit the subscript for the chemical potential from each of the above objects. Furthermore, for any \(f
\in \sphilb{D}_{0,\sminvtemperature}\), we define the sesquilinear form \[\opform{q}_{\txtbec}(f)
=
\opform{q}_{0}(f)
+\opform{q}_{\txtnonzero}(f).\]

\subsection{Weyl algebra and resolvent algebra}\label{expedition0012083}

Let \(\sphilb{H}\) be a complex Hilbert space, and define the symplectic form by the bilinear map \[\sigma
\colon \sphilb{H} \times \sphilb{H}
\to
\fldreal;
\quad
\sigma(f,g)
=
\opimag \bkt{f}{g}_{\sphilb{H}}.\] For any \(f,g
\in \sphilb{H}\), the \(\oacstar\)-algebra \[\oaweyl
=
\oaweyl(\sphilb{H}, \sigma)
=
\oaweyl(\sphilb{H})
=
\oacstar
\set{\opfockweyl(f)}{f \in \sphilb{H}}\]whose generators \(\opfockweyl(f)\) satisfy the Weyl relations \begin{equation}
\begin{aligned}
\faadj{\opfockweyl(f)}
&=
\opfockweyl(-f), \\
\opfockweyl(f) \opfockweyl(g)
&=
\napiernum^{-\frac{\imunit}{2} \opimag \bkt{f}{g}_{\sphilb{H}}}
\opfockweyl(f+g)
\end{aligned}
\end{equation} is called the Weyl algebra. Unless there is risk of confusion, we use appropriate abbreviations for the Weyl algebra as the occasion demands. We may specify an appropriate subspace rather than the full Hilbert space; the Weyl algebra over the full Hilbert space is also called the full Weyl algebra. When we wish to distinguish clearly from the Fock representation or the Araki--Woods representation of the Weyl algebra, the above abstractly defined Weyl algebra is also called the abstract Weyl algebra.

Let \((\sphilb{H}_{\repn},\repn)\) be a representation of the abstract Weyl algebra \(\oaweyl(\sphilb{H})\). When the unitary group \(t
\in \fldreal
\to \repn(\opfockweyl(tf))\) is strongly continuous for every \(f
\in \sphilb{H}\), the representation \((\sphilb{H}_{\repn},\repn)\) is called a regular representation. A state \(\omega\) on the Weyl algebra \(\oaweyl(\sphilb{H})\) is called a regular state when its GNS representation is regular. Furthermore, let \(\oaweyl(\sphilb{D})\) be the Weyl algebra over a pre-Hilbert space \(\sphilb{D}\). When a representation \(\pairbk{\sphilb{H},\repn}\) of the Weyl algebra satisfies that \(\fldreal \ni t
\mapsto \repn(\opfockweyl(tf))\) is strongly continuous for every \(f
\in \sphilb{D}\), the representation is called a regular representation. Furthermore, when the GNS representation of a state \(\oastate\) on the Weyl algebra is a regular representation, the state \(\oastate\) is called a regular state.

Let \(\fun{\spfock_{\txtbsn}}{\sphilb{H}}\) denote the Fock space over the complex Hilbert space \(\sphilb{H}\), and define the Segal field operator and the symplectic form \(\sigma\) by \[\opfocksegal_{\txtfock}(f)
=
\frac{1}{\sqrt{2}}
\rbk{\opfockcr_{\txtfock}(f) + \opfockan_{\txtfock}(f)},
\quad
\sigma_{\txtfock}(f,g)
=
\opimag \bkt{f}{g}_{\sphilb{H}}.\] We then define \[\repn_{\txtfock}
\colon \oaweyl(\sphilb{H},\sigma)
\to \opspbddlin{\fun{\spfock_{\txtbsn}}{\sphilb{H}}};
\quad
\repn_{\txtfock}(\opfockweyl(f))
=
\opfockweyl_{\txtfock}(f)
=
\napiernum^{\imunit \opfocksegal_{\txtfock}(f)}.\] We call this \(\repn_{\txtfock}\) the Fock representation of the abstract Weyl algebra, or simply the Fock representation of the Weyl algebra.

Following \cite{DetlevBuchholz001}, we introduce the definition and basic properties of the resolvent algebra. Let \((X,\sigma)\) be a symplectic space. Let \(\oaresolventalgebra_0\) be the universal unital \(\ast\)-algebra generated by the set \(\set{\oaresolvent(\lambda,f)}
{\lambda \in \fldmultiplicativegroup{\fldreal}, f \in \sphilb{H}}\), subject to the following resolvent relations: \begin{align}
\oaresolvent(\lambda,0)
&=
-\frac{\imunit}{\lambda} \idone, \\ 
\faadj{\oaresolvent(\lambda,f)}
&=
\oaresolvent(-\lambda,f), \\ 
\nu \oaresolvent(\nu \lambda, \nu f)
&=
\oaresolvent(\lambda, f), \\ 
\oaresolvent(\lambda,f) - \oaresolvent(\mu,f)
&=
\imunit
(\mu - \lambda)
\oaresolvent(\lambda,f) \cdot \oaresolvent(\mu,f) \\ 
&=
\imunit
(\mu - \lambda)
\oaresolvent(\mu,f) \cdot \oaresolvent(\lambda,f), \\ 
\commutator{\oaresolvent(\lambda,f)}{\oaresolvent(\mu,g)}
&=
\imunit
\sigma(f,g)
\oaresolvent(\lambda,f)
\oaresolvent(\mu,g)^2
\oaresolvent(\lambda,f), \label{expedition0012052} \\ 
\oaresolvent(\lambda,f)
\oaresolvent(\mu,g)
&=
\oaresolvent(\lambda+\mu, f+g)
\cdot
\rbkleft{\oaresolvent(\lambda,f)} \\
&\quad\rbkright{+
\oaresolvent(\mu,g)
+\imunit \sigma(f,g) \oaresolvent(\lambda,f)^2 \oaresolvent(\mu,g)}. 
\end{align} In particular, by condition \eqref{expedition0012052}, \(\oaresolvent(\lambda,f)\) and \(\oaresolvent(\mu,f)\) with the same \(f\) commute.

The \(\ast\)-algebra obtained by introducing an appropriate norm on \(\oaresolventalgebra_0\) and completing it is called the abstract resolvent algebra, or simply the resolvent algebra. For details on the norm, see \cite[P.2730, Definition 3.4]{BuchholzGrundling2}. In particular, by \cite[P.2730, Theorem 3.6 (iii)]{BuchholzGrundling2}, we have \(\norm{\oaresolvent(\lambda,f)}
= \frac{1}{\abs{\lambda}}\).

As a dense subalgebra, we choose the \(\ast\)-subalgebra generated by finite products of the generators \(\oaresolvent(\lambda,f)\) of \(\oaresolventalgebra(\sphilb{H},\sigma)\); in symbols, \[\oaresolventalgebra_{\txtfin}
=
\oastaralgebra
\set{\prod^{\txtfin} \oaresolvent(z_j,f_j)}{z_j \in \fldcmp \setminus \fldreal, f_j \in X}.\] Furthermore, the \(\ast\)-subalgebra obtained by restricting \(\sphilb{H}\) to an arbitrary subspace \(\sphilb{D}\) is denoted by \(\oaresolventalgebra_{\txtfin}(\sphilb{D},\sigma)\). In the discussion of Bose-Einstein condensation of the free Bose gas and the van Hove model, the first argument can become lengthy and the boundary with the second argument may be unclear; accordingly, we sometimes use a semicolon as a delimiter and write \(\oaresolvent(\lambda;f)\).

As is well known for ordinary resolvents, the resolvent is analytic in the first argument, and the same property holds for the general resolvent algebra. Using this, we extend the real variable \(\lambda
\in \fldreal\) of the resolvent algebra to a complex variable \(z
\in \fldcmp \setminus \imunit \fldreal\), obtaining the relations \begin{align}
\oaresolvent(z,0)
&=
-\frac{\imunit}{z} \idone, \\ 
\faadj{\oaresolvent(z,f)}
&=
\oaresolvent(-\cmpconj{z},f), \\ 
\nu \oaresolvent(\nu z, \nu f)
&=
\oaresolvent(z,f),
\quad
\nu \in \fldmultiplicativegroup{\fldreal}, \\ 
\oaresolvent(z,f) - \oaresolvent(w,f)
&=
\imunit
(w - z)
\oaresolvent(z,f) \cdot \oaresolvent(w,f) \\ 
&=
\imunit
(w - z)
\oaresolvent(w,f) \cdot \oaresolvent(z,f), \\ 
\commutator{\oaresolvent(z,f)}{\oaresolvent(w,g)}
&=
\imunit
\sigma(f,g)
\oaresolvent(z,f)
\oaresolvent(w,g)^2
\oaresolvent(z,f), \\ 
\oaresolvent(z,f)
\oaresolvent(w,g)
&=
\oaresolvent(z+w, f+g)
\cdot
\rbkleft{\oaresolvent(z,f)} \\
&\quad\rbkright{+
\oaresolvent(w,g)
+\imunit \sigma(f,g) \oaresolvent(z,f)^2 \oaresolvent(w,g)} 
\end{align} which are also called the resolvent relations.

Let \(\oaresolventalgebra(X,\sigma)\) be the resolvent algebra, and let \(S\) be a subset of the symplectic space \(X\). A representation \(\oarepn
\in \Rep(\oaresolventalgebra(X,\sigma), \sphilb{H}_{\oarepn})\) is called a regular representation on \(S\) if \(\Ker \oarepn(\oaresolvent(1,f))
= \setone{0}\) holds for every \(f
\in S\). A state \(\oastate\) on the resolvent algebra is called a regular state when its GNS representation is a regular representation on \(X\).

\begin{prop}[\cite{BuchholzGrundling2}]\label{expedition0011838}
Let $\pairbk{X,\sigma}$ be a symplectic space of arbitrary dimension, and let $S
\subset X$ be a non-degenerate finite-dimensional subspace.
\begin{enumerate}
\item
The norms of the full resolvent algebra $\oaresolventalgebra(X,\sigma)$ and the subalgebra $\oaresolventalgebra(X,\sigma)$ coincide on the $\ast$-subalgebra $$\oastaralgebra
\set{\oaresolvent(\lambda,f)}
{f \in S, \lambda \in \fldreal \setminus \setone{0}}.$$
In particular, $\oaresolventalgebra(S,\sigma)
\subset \oaresolventalgebra(X,\sigma)$ holds.

\item
The full resolvent algebra is the inductive limit of the net $\fml{\oaresolventalgebra(S,\sigma)}
{X \subset S}$ over non-degenerate finite-dimensional subspaces $S
\subset X$.

\item
Every regular representation of the full resolvent algebra $\oaresolventalgebra(X,\sigma)$ is faithful.
\end{enumerate}
In particular, the center of the full resolvent algebra is trivial.
\end{prop}

\subsection{Theorems}\label{theorems}

We define the Segal field operator satisfying the self-consistency condition by \(\opfocksegal_{\txtselfconsistent}(f)
=
\opfocksegal(f)
+\opreal \mathsf{m}(f)\), where \(\mathsf{m}(f)
= \opreal \bkt{\varrho / \omega^{3/2}}{f}\).

\begin{thm}[Selection criterion for physical phonons, Theorem \ref{expedition0012077}]
For the $\sminvtemperature$-KMS state $\oastate[\psi_{\txtvanhowe,\sminvtemperature}]$ of the van Hove model, we have $\fun{\oastate[\psi_{\txtvanhowe,\sminvtemperature}]}
{\opfocksegal_{\txtselfconsistent}(f)}
= 0$.
\end{thm}

This condition corresponds to the self-consistency condition of \cite{VIYukalov001}. However, the self-consistency condition is a constraint concerning the definition of the field, and this alone is insufficient to establish the no-go theorem for BEC: see Remark \ref{expedition0012090}.

For the \(\sminvtemperature\)-KMS states of the van Hove model, the following expression is obtained.

\begin{thm}[Theorem \ref{expedition0011612}, Theorem \ref{expedition0011637}]\label{expedition0012087}
The expectation value of the Weyl operator in the $\sminvtemperature$-KMS state $\oastate[\psi_{\txtvanhowe,\sminvtemperature}]$ for the van Hove model is $$\oastate[\psi_{\txtvanhowe,\sminvtemperature}]
(\opfockweyl(f))
=
\fnexp{-\imunit
\opreal \mathsf{m}(f)
-\oneoverfour
\opform{q}_{\txtnonzero}(f)
-\oneoverfour
\opform{q}_{0}(f)}.$$
\end{thm}

By this expression, the occurrence of BEC is controlled by the sesquilinear form \(\opform{q}_{0}\). The no-go theorem for BEC is obtained via the following two routes.

First, when the degree of the dispersion relation is high, the treatment of infrared divergences reduces the algebra of physical observables, and BEC is mathematically excluded.

\begin{thm}
Let the dispersion relation be $\omega_s(k)
= \abs{k}^s$ for $s > 2$.
On the algebra of physical observables incorporating the infrared singularity condition,
$\opform{q}_0(f)
= 0$ holds for every $f
\in \sphilb{D}_{\txtirsingular,\sminvtemperature}$.
In particular, BEC does not occur.
\end{thm}

If the higher-order dispersion relation is renormalized into the background field in order to suppress the lattice collapse caused by infrared divergences, the reduction of the algebra of physical observables corresponds to this physical picture.

On the other hand, for \(1
\leq s \leq 2\), in particular for the linear dispersion \(s
= 1\), this theorem does not apply, and a physical assumption is required for the no-go theorem for BEC. We formulate the mildness of the equilibrium state for phonons in \cite{VIYukalov001} in terms of time cluster properties. By Theorem \ref{expedition0012087}, BEC in the van Hove model is described in the same form as in the free Bose gas, so the results of \cite{AsaoArai28,YoshitsuguSekine004} are applicable.

\begin{thm}
Suppose that the $\sminvtemperature$-KMS state of the van Hove model satisfies the time cluster property.
Then the van Hove model does not exhibit BEC.
\end{thm}

By \cite{AsaoArai26,YoshitsuguSekine004}, for the equilibrium states of the free Bose gas, the validity of the time cluster property is equivalent to the validity of the space cluster property. Since the space cluster property also represents spatial mildness for the equilibrium state of phonons, the cluster property is a natural formulation of the mildness of the equilibrium state assumed in \cite{VIYukalov001}.

\section{Basic setup for the van Hove model}\label{expedition0011236}

\subsection{\texorpdfstring{The source \(\varrho\) and its cutoffs \(\varrho_{\kappa},\varrho_{\Lambda},\varrho_{\kappa,\Lambda}\)}{The source \textbackslash varrho and its cutoffs \textbackslash varrho\_\{\textbackslash kappa\},\textbackslash varrho\_\{\textbackslash Lambda\},\textbackslash varrho\_\{\textbackslash kappa,\textbackslash Lambda\}}}\label{the-source-varrho-and-its-cutoffs-varrho_kappavarrho_lambdavarrho_kappalambda}

The source \(\varrho\) representing the interaction term is a Dirac delta function with weight at the origin, that is, \(\varrho
= \diracdelta_0\), and the source with infrared and ultraviolet cutoffs \(\varrho_{\kappa,\Lambda}\) is defined by \[\faftr{\varrho_{\kappa,\Lambda}}(k)
=
\faftr{\varrho}(k) \fndef{\kappa \leq \abs{k} \leq \Lambda}(k),
\quad
\faftr{\varrho}(k)
=
\faftr{\diracdelta_0}(k)
=
1.\] In particular, when the infrared cutoff is removed by setting \(\kappa
= 0\), we write \(\varrho_{\Lambda}\), and when the ultraviolet cutoff is removed by setting \(\Lambda
= \infty\), we write \(\varrho_{\kappa}\).

Then for any \(\kappa
> 0\), the infrared regularization condition \(\varrho_{\kappa,\Lambda}
\in \dom \omega^{-\frac{3}{2}}\), namely \[\int_{\greuctr{k}{3}}
\frac{\abs{\faftr{\varrho_{\kappa,\Lambda}}(k)}^2}{\omega(k)^3}
\opdmsr{k}
=
\int_{\kappa \leq \abs{k} \leq \Lambda}
\frac{1}{\omega(k)^3}
\opdmsr{k}
<
\infty\]holds, together with the convergence in the topology of the space of tempered distributions \(\fun{\dsttempered}{\fldreal^{3}}\): \[\lim_{\kappa \to 0, \Lambda \to \infty}
\varrho_{\kappa,\Lambda}
=
\varrho.\] In particular, the source with cutoffs satisfies the condition \(\varrho
\in
\dom \omega^{-1}\) around the origin, that is, the integrability condition for any \(\kappa
> 0\): \[\int_{\abs{k} \leq \kappa}
\frac{\abs{\faftr{\varrho_{\kappa,\Lambda}}(k)}^2}{\omega(k)^2}
\opdmsr{k}
<
\infty.\]

We define the mean functional \(\mathsf{m}_{\kappa,\Lambda}, \mathsf{m}\) using the source and the domain \(\dom \mathsf{m}\) as follows.

\begin{defn}\label{expedition0011100}
So that they make sense as functions on the wave-number space via Fourier transform, we define $$\mathsf{m}_{\kappa,\Lambda}
=
\omega^{-\frac{3}{2}} \varrho_{\kappa,\Lambda},
\quad
\mathsf{m}
=
\omega^{-\frac{3}{2}} \varrho$$and denote the induced mean functionals by the same symbols.
Specifically,
\begin{equation}
\begin{aligned}
\mathsf{m}_{\kappa,\Lambda}(f)
&=
\bkt{\omega^{-\frac{3}{2}} \varrho_{\kappa,\Lambda}}{f}_{\sphilb{H}}, \\
\mathsf{m}(f)
&=
\int_{\greuctr{k}{3}}
\frac{\cmpconj{\faftr{\varrho}(k)} \faftr{f}(k)}{\omega(k)^{\frac{3}{2}}}
\opdmsr{k}.
\end{aligned}
\end{equation}
The domain of the former is the entire $\sphilb{H}$,
and the domain of the latter is $\dom \mathsf{m}$.
\end{defn}

Since the \(\lp^{2}\)-norm of the function \(\mathsf{m}\) without cutoff is \[\norm{\mathsf{m}}_{\fun{\lp^{2}}{\fldreal^{3}}}^2
=
\int_{\greuctr{k}{3}}
\frac{\abs{\faftr{\varrho}(k)}^2}{\omega(k)^3}
\opdmsr{k}\]the integrability of \(\mathsf{m}\) as a function can also be used for determining infrared singularity.

As discussed in the formulation via operator theory in \cite{AsaoArai26} (using different notation), the square integrability of \(\mathsf{m}\) as a function around the origin in wave-number space is a property that controls infrared divergence. When infrared (and ultraviolet) cutoffs are present, this square integrability is always guaranteed; otherwise, that is, when infrared divergence occurs, the condition imposing constraints on observable physical quantities is the specification of the subspace \(\dom \mathsf{m}\).

\begin{rem}
The better the properties of the source $\varrho$, the larger $\dom \mathsf{m}$ can be taken and the more physical quantities can be observed;
conversely, the worse the properties of $\varrho$, the smaller $\dom \mathsf{m}$ becomes,
and the fewer physical quantities can be observed without being affected by divergences.
As the notation indicates, this is also influenced by $\omega$
and by the spatial dimension.
Here we consider the domain of the functional $\mathsf{m}$ that appears in the definition of the automorphism group.

If the source $\varrho$ is a rapidly decreasing function and $\omega(k)
= \abs{k}$,
then in the infrared region there is no singularity around the origin,
and there is no need to incorporate the effect into $\dom \mathsf{m}$.
This is because of $\abs{k}^{-\frac{3}{2}}$ in the numerator and the $k^2$ arising from the Jacobian of the change of variables.
In the ultraviolet region as well, the rapid decrease of $\varrho$ suppresses the divergence at infinity,
and again there is no need to incorporate the effect into $\dom \mathsf{m}$.

If we take $\omega(k)
= \abs{k}$ and $\varrho$ to be a point source,
then as before,
in the infrared region there is no effect on $\dom \mathsf{m}$.
However, in the ultraviolet region, in order to suppress the divergence due to $\abs{k}^{\onehalf}$,
ultraviolet integrability such as $$\int_{\abs{k} \geq 1}
\abs{k}^{\onehalf}
\faftr{f}(k)
\opdmsr{k}
<
\infty$$is required, which affects $\dom \mathsf{m}$.

Next, suppose $\varrho$ is a point source and the dispersion relation is $\omega(k)
= k^{2+\delta}$ for $\delta
> 0$.
In this case, in the infrared region, $\abs{k}^{-3}$ in the numerator leaves $\abs{k}^{-1}$ even after cancellation by the Jacobian of the change of variables,
and this affects $\dom \mathsf{m}$ in order to avoid logarithmic divergence.
In particular, the behavior $\faftr{f}(k)
= \fun{O}{\abs{k}^{1+\delta}}$ as $k
\to 0$ is required,
and in particular $\faftr{f}(0)
= 0$ is demanded.
Conversely, in the ultraviolet region this contributes as a damping term,
affecting $\dom \mathsf{m}$ in the form of reducing constraints on its elements.
\end{rem}

For brevity of notation, we define auxiliary functionals \(\mathsf{M}_{\kappa,\Lambda,t}\) and \(\mathsf{M}_{t}\).

\begin{defn}\label{expedition0011628}
We define auxiliary functionals by
\begin{equation}
\begin{aligned}
\mathsf{M}_{\kappa,\Lambda,t}(f)
&=
\opreal \fun{\mathsf{m}_{\kappa,\Lambda}}{\rbk{\napiernum^{\imunit t \omega} - \idone} f}, \\
\mathsf{M}_{t}(f)
&=
\opreal \fun{\mathsf{m}}{\rbk{\napiernum^{\imunit t \omega} - \idone} f}.
\end{aligned}
\end{equation}
\end{defn}

\begin{prop}\label{expedition0011238}
The auxiliary functionals $\mathsf{M}_{\kappa,\Lambda,t}$ and $\mathsf{M}_{t}$ satisfy the cocycle condition
\begin{equation}
\begin{aligned}
\mathsf{M}_{\kappa,\Lambda,t+s}(f)
&=
\mathsf{M}_{\kappa,\Lambda,s}(\napiernum^{\imunit t \omega} f) + \mathsf{M}_{\kappa,\Lambda,t}(f), \\
\mathsf{M}_{t+s}(f)
&=
\mathsf{M}_{s}(\napiernum^{\imunit t \omega} f) + \mathsf{M}_{t}(f).
\end{aligned}
\end{equation}
\end{prop}

\begin{proof}
Since the argument is the same for either, we work with $\mathsf{M}_{t}(f)$, which has simpler notation.
A direct computation yields
\begin{equation}
\begin{aligned}
&\mathsf{M}_{t+s}(f)
=
\opreal \bkt{\mathsf{m}}{\rbk{\napiernum^{\imunit (t+s) \omega} - 1} f} \\
&=
\opreal \bkt{\mathsf{m}}{\rbk{\napiernum^{\imunit s \omega} - 1} \napiernum^{\imunit t \omega} f}
+\opreal \bkt{\mathsf{m}}{\rbk{\napiernum^{\imunit t \omega} - 1} f} \\
&=
\mathsf{M}_{s}(\napiernum^{\imunit t \omega} f)
+\mathsf{M}_{t}(f).
\end{aligned}
\end{equation}
\end{proof}

We now define the concrete van Hove model with infrared and ultraviolet cutoffs on Fock space by \[\physham_{\txtvanhowe,\kappa,\Lambda}
=
\physham_{\txtbsn,\txtfr}
+\fun{\opfocksegal_{\txtfock}}{\frac{\varrho_{\kappa,\Lambda}}{\omega^{\onehalf}}}
=
\physham_{\txtbsn,\txtfr}
+\fun{\opfocksegal_{\txtfock}}{\omega \mathsf{m}_{\kappa,\Lambda}}.\] By the Kato--Rellich theorem, this is also a self-adjoint operator bounded from below.

\subsection{Basic setup for finite temperature}\label{expedition0011278}

Here again we use the setup adopted in \cite{YoshitsuguSekine004}. Let \(I_L
= \closedinterval{-\frac{L}{2}}{\frac{L}{2}}\) be the closed interval of side length \(L
> 0\), and define the Hilbert space of one-particle states moving in the hypercube \(I_{L}^{3}\) by \[\fun{\lp^{2}}{I_{L}^{3}}
=
\set{f \in \fun{\mblfn_{\txtborel}}{I_{L}^{3};\fldcmp}}
{\int_{I_{L}^{3}} \abs{f(x)}^2 \opdmsr{x} < \infty}\]and the lattice for the bounded system in wave-number space by \[\setlattice_L
=
\frac{2 \pi}{L} \ringratint 
=
\set{\frac{2 \pi}{L} n}{n \in \ringratint}.\] We impose periodic boundary conditions on \(I_{L}^{3}\) for the one-particle Hamiltonian or dispersion relation \(\omega\). For a positive real number \(\smchemicalpotential
> 0\), we set \(\physham_{\txtbsn,\txtfr}(\smchemicalpotential)
= \fun{\opfocksndqntdiff_{\txtbsn}}{\omega - \smchemicalpotential}\), and as an operator on the bosonic Fock space \(\fun{\spfock_{\txtbsn}}{\fun{\lp^{2}}{I_L^d}}\), we define the Hamiltonian of the van Hove model for the bounded system with chemical potential by \[\physham_{\txtvanhowe,I_{L}^{3},\kappa,\Lambda}(\smchemicalpotential)
=
\physham_{\txtbsn,\txtfr}(\smchemicalpotential)
+\fun{\opfocksegal_{\txtfock}}{\omega \mathsf{m}_{\kappa,\Lambda}}.\] By the Kato--Rellich theorem, this is a self-adjoint operator bounded from below.

The automorphism group \(\alpha_{\txtvanhowe,I_{L}^{3},\kappa,\Lambda,\smchemicalpotential}\) generated by this Hamiltonian is defined by \[\alpha_{\txtvanhowe,I_{L}^{3},\kappa,\Lambda,\smchemicalpotential,t}
=
\Ad U_{I_{L}^{3},\kappa,\Lambda,\smchemicalpotential},
\quad
U_{I_{L}^{3},\kappa,\Lambda,\smchemicalpotential}
=
\napiernum^{\imunit t \physham_{\txtvanhowe,I_{L}^{3},\kappa,\Lambda}(\smchemicalpotential)}.\] For inverse temperature \(\sminvtemperature
> 0\), the density operator \(\opdmat_{\txtvanhowe,I_{L}^{3},\kappa,\Lambda,\sminvtemperature,\smchemicalpotential}\) and the grand partition function \(\smgrandpartitionfunc_{\txtvanhowe,I_{L}^{3},\kappa,\Lambda,\sminvtemperature,\smchemicalpotential}\) are defined by \begin{equation}
\begin{aligned}
\opdmat_{\txtvanhowe,I_{L}^{3},\kappa,\Lambda,\sminvtemperature,\smchemicalpotential}
&=
\frac{1}
{\smgrandpartitionfunc_{\txtvanhowe,I_{L}^{3},\kappa,\Lambda,\sminvtemperature,\smchemicalpotential}}
\napiernum^{-\sminvtemperature \physham_{\txtvanhowe,I_{L}^{3},\kappa,\Lambda}(\smchemicalpotential)},
\\ 
\smgrandpartitionfunc_{\txtvanhowe,I_{L}^{3},\kappa,\Lambda,\sminvtemperature,\smchemicalpotential}
&=
\sqfun{\trace}{\napiernum^{-\sminvtemperature
\physham_{\txtvanhowe,I_{L}^{3},\kappa,\Lambda}(\smchemicalpotential)}}.
\end{aligned}
\end{equation} For the infinite system, the same definitions are made using notation with \(I_{L}^{3}\) removed. When the chemical potential is \(0\), we simply omit the notation for the chemical potential.

For the local density operator \[\smlocaldensityoperator_{\sminvtemperature,\smchemicalpotential}
=
\frac{1}{\napiernum^{\sminvtemperature (\omega - \smchemicalpotential)} - 1}\]the operator \(K_{\sminvtemperature,\smchemicalpotential}\) is defined by \[K_{\sminvtemperature,\smchemicalpotential}
=
2 \smlocaldensityoperator_{\sminvtemperature,\smchemicalpotential} + 1
=
\frac{1 + \napiernum^{-\sminvtemperature (\omega - \smchemicalpotential)}}
{1 - \napiernum^{-\sminvtemperature (\omega - \smchemicalpotential)}}.\] When the chemical potential is \(0\), we simply omit the notation for the chemical potential and write \(\smlocaldensityoperator_{\sminvtemperature}\) or \(K_{\sminvtemperature}\). In systems where the chemical potential is set to \(0\), the algebra of physical quantities is constrained. In particular, the one-particle subspace based on the system with infrared and ultraviolet divergences removed is defined by \[\sphilb{D}_{\txtirsingular,\sminvtemperature}
=
\dom \mathsf{m} \cap \sphilb{D}_{\sminvtemperature}.\] Under the above setup, the algebra of physical quantities as a Weyl algebra is \[\oaweyl_{\txtirsingular,\sminvtemperature}
=
\oaweyl(\sphilb{D}_{\txtirsingular,\sminvtemperature})
=
\oacstar
\set{\opfockweyl(f)}
{f \in \sphilb{D}_{\txtirsingular,\sminvtemperature}}.\]

\section{Bounded systems on the Weyl algebra}\label{expedition0011587}

As discussed in \cite{YoshitsuguSekine002}, the problem essentially reduces to the free Bose gas. However, since there are few references that discuss the van Hove model at finite temperature, we include the basic arguments for the reader's convenience. The basic strategy is the same as the ground state discussion in \cite{AsaoArai26}, namely a unitary transformation by Weyl operators under the assumption of infrared regularization.

Since the computational results are concise, we first compute various quantities for bounded systems at finite temperature in the Fock representation of the Weyl algebra. In particular, we work under the assumption of infrared and ultraviolet cutoffs in the Fock representation. In this section, since we work with nonzero chemical potential under infrared and ultraviolet regularization, the Hilbert space generating the Weyl algebra is the entire \(\sphilb{H}_L
= \fun{\lp^{2}}{I_{L}^{3}}\).

The following statement is obvious by the Fourier transform.

\begin{prop}\label{expedition0011286}
The one-particle Hamiltonian or dispersion relation $\omega$ is a non-negative self-adjoint operator on $\fun{\lp^{2}}{I_{L}^{3}}$.
Its spectrum is discrete and satisfies $$\opspec{\omega}
=
\opspec[\txtdiscrete]{\omega}
=
\set{\omega(k)}{k \in \setlattice_L^3}.$$
\end{prop}

\begin{prop}
For any inverse temperature $\sminvtemperature
> 0$, we have $$\sqfun{\trace_{\fun{\lp^{2}}{I_{L}^{3}}}}{\napiernum^{-\sminvtemperature \omega}}
< \infty$$
\end{prop}

\begin{proof}
Since the spectrum of the dispersion relation $\omega$ is discrete for the bounded system, we have $$S
=
\sqfun{\trace_{\fun{\lp^{2}}{I_{L}^{3}}}}{\napiernum^{-\sminvtemperature \omega}}
=
\sum_{k \in \setlattice_L^3} \napiernum^{-\sminvtemperature \omega(k)}.$$
It remains to verify that the series on the right-hand side converges.

For brevity, set $a
= \frac{2 \pi}{L}$.
Considering a shell decomposition by radius, we obtain $$S
\leq
1
+\sum_{m=0}^{\infty}
N_m \napiernum^{-\sminvtemperature m},
\quad
N_m
=
\setcardop \set{n \in \ringratint^3}{m \leq \abs{n} \leq m+1}.$$
By volume comparison, the number of lattice points can be estimated by constants $C,C'$ depending only on the dimension: $$N_m
\leq
C \rbk{\rbk{m+1}^3 - m^3}
\leq
C' \rbk{m^2 + 1}.$$
Therefore, using an appropriate positive constant $C''$, we obtain $$S
\leq
1 + C'' \sum_{m=0}^{\infty} (m^2+1) \napiernum^{-\sminvtemperature a m}
<
\infty.$$
\end{proof}

We define the unitary transformation introduced in \cite{AsaoArai26} by \(V_{\kappa,\Lambda}
= \opfockweyl_{\txtfock}(\imunit \mathsf{m}_{\kappa,\Lambda})\). Essentially the same transformation is used in the discussion of the Hubbard--phonon interaction system \cite{YoshitsuguSekine001,YoshitsuguSekine002}. For various computations, we cite from \cite{AsaoArai26,AsaoArai28} basic properties of Weyl operators and results on commutation relations with creation and annihilation operators.

\begin{fact}\label{expedition0010960}
Let $f,g$ be arbitrary elements of the Hilbert space $\sphilb{H}$.
\begin{enumerate}
\item
The Weyl operator leaves the domain of creation and annihilation operators invariant.
That is, $$\napiernum^{\imunit \opfocksegal(f)} \opfockcran(g)
= \opfockcran(g)$$holds.
Furthermore, the operator identities
\begin{align}
\napiernum^{\imunit \opfocksegal(f)}
\opfockan(g)
\napiernum^{-\imunit \opfocksegal(f)}
&=
\opfockan(g)
-\frac{\imunit}{\sqrt{2}}
\bkt{g}{f}, \\ 
\napiernum^{\imunit \opfocksegal(f)}
\opfockcr(g)
\napiernum^{-\imunit \opfocksegal(f)}
&=
\opfockcr(g)
+\frac{\imunit}{\sqrt{2}}
\bkt{f}{g} 
\end{align}
hold.

\item
The Weyl operator leaves the domain of the Segal field operator invariant.
That is, $$\napiernum^{\imunit \opfocksegal(f)} \dom \opfocksegal(g)
= \dom \opfocksegal(g)$$holds.
Furthermore, the operator identity
\begin{align}
\napiernum^{\imunit \opfocksegal(f)}
\opfocksegal(g)
\napiernum^{-\imunit \opfocksegal(f)}
=
\opfocksegal(g)
-\opimag \bkt{f}{g} 
\end{align}
holds.
\end{enumerate}
\end{fact}

\begin{prop}\label{expedition0012088}
The Hamiltonian $\physham_{\txtvanhowe,I_{L}^{3},\kappa,\Lambda}(\smchemicalpotential)$ of the van Hove model is unitarily transformed by the operator $V_{\kappa,\Lambda}
= \opfockweyl_{\txtfock}(\imunit \mathsf{m}_{\kappa,\Lambda})$ as $$V_{\kappa,\Lambda}
\physham_{\txtbsn,\txtfr}(\smchemicalpotential)
\inv{V_{\kappa,\Lambda}}
=
\physham_{\txtvanhowe,I_{L}^{3},\kappa,\Lambda}(\smchemicalpotential)
-\fun{\physgse}{\physham_{\txtvanhowe,I_{L}^{3},\kappa,\Lambda}(\smchemicalpotential)}.$$
Here $\fun{\physgse}{\physham_{\txtvanhowe,I_{L}^{3},\kappa,\Lambda}(\smchemicalpotential)}$ is the ground state energy.
\end{prop}

\begin{proof}
It suffices to use Fact \ref{expedition0010960} together with the result on the unitary transformation of the free Hamiltonian by the exponential of the Segal field operator.
The computations in \cite{YoshitsuguSekine001,YoshitsuguSekine002} may also serve as a reference.
\end{proof}

By the following proposition, the van Hove Hamiltonian with chemical potential is trace class in the bounded system.

\begin{prop}\label{expedition0011287}
The heat operator $\napiernum^{-\sminvtemperature \physham_{\txtvanhowe,I_{L}^{3},\kappa,\Lambda}(\smchemicalpotential)}$ of the van Hove Hamiltonian with chemical potential is trace class, and in particular satisfies $$\sqfun{\trace}{\napiernum^{-\sminvtemperature \physham_{\txtvanhowe,I_{L}^{3},\kappa,\Lambda}(\smchemicalpotential)}}
=
\frac{1}
{\napiernum^{\sminvtemperature
\fun{\physgse}
{\physham_{\txtvanhowe,I_{L}^{3},\kappa,\Lambda}(\smchemicalpotential)}}}
\sqfun{\trace}
{\napiernum^{-\sminvtemperature \physham_{\txtbsn,\txtfr}(\smchemicalpotential)}}.$$
\end{prop}

\begin{proof}
As is well known, the partition function of the free Bose gas is $$\sqfun{\trace}
{\napiernum^{-\sminvtemperature \physham_{\txtbsn,\txtfr}(\smchemicalpotential)}}
=
\frac{1}
{\prod_{k \in \setlattice_L^3}
\rbk{1 - \napiernum^{-\sminvtemperature(\omega(k) - \smchemicalpotential)}}}
<
\infty.$$
Since the spectrum is preserved under unitary transformation,
the heat operator of $$\physham_{\txtvanhowe,I_{L}^{3},\kappa,\Lambda}(\smchemicalpotential)
-\fun{\physgse}
{\physham_{\txtvanhowe,I_{L}^{3},\kappa,\Lambda}
(\smchemicalpotential)}$$is also trace class, and we obtain $$\sqfun{\trace}
{\napiernum^{-\sminvtemperature
\physham_{\txtvanhowe,I_{L}^{3},\kappa,\Lambda}
(\smchemicalpotential)}}
=
\napiernum^{-\sminvtemperature
\fun{\physgse}
{\physham_{\txtvanhowe,I_{L}^{3},\kappa,\Lambda}
(\smchemicalpotential)}}
\sqfun{\trace}
{\napiernum^{-\sminvtemperature
\physham_{\txtbsn,\txtfr}(\smchemicalpotential)}}.$$
\end{proof}

Using the above computational results, we define the grand partition function and the grand canonical state. Specifically, \begin{equation}
\begin{aligned}
\smgrandpartitionfunc_{\txtbsn,\txtfr,I_{L}^{3},\sminvtemperature,\smchemicalpotential}
&=
\sqfun{\trace}{\napiernum^{-\sminvtemperature \physham_{\txtbsn,\txtfr}(\smchemicalpotential)}}
=
\frac{1}{\prod_{k \in I_{L}^{3}} \rbk{1 - \napiernum^{-\sminvtemperature(\omega(k) - \smchemicalpotential)}}}, \\
\smgrandpartitionfunc_{\txtvanhowe,I_{L}^{3},\kappa,\Lambda,\sminvtemperature,\smchemicalpotential}
&=
\frac{1}
{\napiernum^{\sminvtemperature
\fun{\physgse}
{\physham_{\txtvanhowe,I_{L}^{3},\kappa,\Lambda}(\smchemicalpotential)}}}
\smgrandpartitionfunc_{\txtbsn,\txtfr,I_{L}^{3},\sminvtemperature,\smchemicalpotential}.
\end{aligned}
\end{equation}

The latter \(\smgrandpartitionfunc
_{\txtvanhowe,I_{L}^{3},\kappa,\Lambda,\sminvtemperature,\smchemicalpotential}\) is called the grand partition function for the van Hove Hamiltonian. Furthermore, as a state on the algebra of all bounded operators \(\opspbddlin{\fun{\spfock_{\txtbsn}}
{\fun{\lp^{2}}{I_{L}^{3}}}}\), \[\oastate[\psi_{\txtgrandcanonical,\txtvanhowe,I_{L}^{3},\kappa,\Lambda,\sminvtemperature,\smchemicalpotential}](A)
=
\frac{1}{\smgrandpartitionfunc_{\txtbsn,\txtvanhowe,I_{L}^{3},\kappa,\Lambda,\sminvtemperature,\smchemicalpotential}}
\sqfun{\trace}{\napiernum^{-\sminvtemperature \physham_{\txtvanhowe,I_{L}^{3},\kappa,\Lambda}(\smchemicalpotential)} A}\]is defined. This is called the grand canonical state for the van Hove Hamiltonian.

As in the discussion of the free Bose gas in \cite{AsaoArai28}, the domain of the grand canonical state is extended to the unbounded Fock canonical commutation relation algebra containing creation and annihilation operators. We proceed under this assumption in the following.

\begin{prop}\label{expedition0011296}
The one-point functions for the creation operator, annihilation operator, and Segal field operator are given by
\begin{equation}
\begin{aligned}
\oastate[\psi_{\txtgrandcanonical,\txtvanhowe,I_{L}^{3},\kappa,\Lambda,\sminvtemperature,\smchemicalpotential}](\opfockcr(f))
&=
-\frac{1}{\sqrt{2}}
\bkt{\mathsf{m}_{\kappa,\Lambda}}{f}, \\
\oastate[\psi_{\txtgrandcanonical,\txtvanhowe,I_{L}^{3},\kappa,\Lambda,\sminvtemperature,\smchemicalpotential}](\opfockan(f))
&=
-\frac{1}{\sqrt{2}}
\cmpconj{\bkt{\mathsf{m}_{\kappa,\Lambda}}{f}}, \\
\oastate[\psi_{\txtgrandcanonical,\txtvanhowe,I_{L}^{3},\kappa,\Lambda,\sminvtemperature,\smchemicalpotential}](\opfocksegal_{\txtfock}(f))
&=
-\opreal
\bkt{\mathsf{m}_{\kappa,\Lambda}}{f}.
\end{aligned}
\end{equation}

\end{prop}

\begin{proof}
We argue in a somewhat simplified and formal manner.
For brevity, let $\physgse$ denote the ground state energy of the van Hove Hamiltonian,
and set $V_{\kappa,\Lambda}
= \opfockweyl_{\txtfock}(\imunit \mathsf{m}_{\kappa,\Lambda})$.

(Creation operator): Using Proposition \ref{expedition0012088} for the Hamiltonian
and Fact \ref{expedition0010960} for the creation and annihilation operators, we obtain
\begin{equation}
\begin{aligned}
&\sqfun{\trace}
{\napiernum^{-\sminvtemperature \physham_{\txtvanhowe,I_{L}^{3},\kappa,\Lambda}(\smchemicalpotential)}
\cdot
\opfockcr(f)} \\
&=
\sqfun{\trace}
{V_{\kappa,\Lambda}
\inv{V_{\kappa,\Lambda}}
\napiernum^{-\sminvtemperature \physham_{\txtvanhowe,I_{L}^{3},\kappa,\Lambda}(\smchemicalpotential)}
V_{\kappa,\Lambda}
\inv{V_{\kappa,\Lambda}}
\opfockcr(f)} \\
&=
\napiernum^{-\sminvtemperature \physgse}
\sqfun{\trace}
{\napiernum^{-\sminvtemperature \physham_{\txtbsn,\txtfr,I_{L}^{3}}(\smchemicalpotential)}
\inv{V_{\kappa,\Lambda}}
\opfockcr(f)
V_{\kappa,\Lambda}} \\
&=
\napiernum^{-\sminvtemperature \physgse}
\sqfun{\trace}
{\napiernum^{-\sminvtemperature \physham_{\txtbsn,\txtfr,I_{L}^{3}}(\smchemicalpotential)}
\rbk{\opfockcr(f) + \frac{\imunit}{\sqrt{2}} \bkt{-\imunit \mathsf{m}_{\kappa,\Lambda}}{f}}} \\
&=
-\frac{1}{\sqrt{2}}
\napiernum^{-\sminvtemperature \physgse}
\bkt{\mathsf{m}_{\kappa,\Lambda}}{f}
\sqfun{\trace}
{\napiernum^{-\sminvtemperature \physham_{\txtbsn,\txtfr,I_{L}^{3}}(\smchemicalpotential)}} \\
&=
-\frac{1}{\sqrt{2}}
\bkt{\mathsf{m}_{\kappa,\Lambda}}{f}
\smgrandpartitionfunc_{\txtbsn,\txtfr,I_{L}^{3},\sminvtemperature,\smchemicalpotential}.
\end{aligned}
\end{equation}
The desired result then follows from the evaluation of the grand partition function.

(Annihilation operator): The only difference is in the application of Fact \ref{expedition0010960} to the commutation relations;
specifically,
\begin{equation}
\begin{aligned}
&\sqfun{\trace}
{\napiernum^{-\sminvtemperature \physham_{\txtvanhowe,I_{L}^{3},\kappa,\Lambda}(\smchemicalpotential)}
\cdot
\opfockan(f)} \\
&=
\napiernum^{-\sminvtemperature \physgse}
\sqfun{\trace}
{\napiernum^{-\sminvtemperature \physham_{\txtbsn,\txtfr,I_{L}^{3}}(\smchemicalpotential)}
\rbk{\opfockan(f) - \frac{\imunit}{\sqrt{2}} \bkt{f}{-\imunit \mathsf{m}_{\kappa,\Lambda}}}} \\
&=
-\frac{1}{\sqrt{2}}
\bkt{f}{\mathsf{m}_{\kappa,\Lambda}}
\napiernum^{-\sminvtemperature \physgse}
\sqfun{\trace}
{\napiernum^{-\sminvtemperature \physham_{\txtbsn,\txtfr,I_{L}^{3}}(\smchemicalpotential)}}.
\end{aligned}
\end{equation}
The desired result again follows from the evaluation of the grand partition function.

(Segal field operator): It suffices to take $\frac{1}{\sqrt{2}}$ times the sum over the creation and annihilation operators.
\end{proof}

We next compute the two-point functions for the grand canonical state.

\begin{prop}\label{expedition0011299}
The two-point functions for the creation operator, annihilation operator, and Segal field operator are given as follows.
Using fraction notation for clarity by exploiting commutativity,
\begin{equation}
\begin{aligned}
&\oastate[\psi_{\txtgrandcanonical,\txtvanhowe,I_{L}^{3},\kappa,\Lambda,\sminvtemperature,\smchemicalpotential}](\opfockcr(f) \opfockan(g)) \\
&=
\bkt{g}
{\frac{\napiernum^{-\sminvtemperature (\omega - \smchemicalpotential)}}
{1 - \napiernum^{-\sminvtemperature (\omega - \smchemicalpotential)}} f}
+\onehalf
\bkt{\mathsf{m}_{\kappa,\Lambda}}{f}
\bkt{g}{\mathsf{m}_{\kappa,\Lambda}}, \\
&=
\bkt{g}
{\frac{1}
{\napiernum^{\sminvtemperature (\omega - \smchemicalpotential)} - 1}
f}
+\onehalf
\bkt{\mathsf{m}_{\kappa,\Lambda}}{f}
\bkt{g}{\mathsf{m}_{\kappa,\Lambda}}, \\
&\oastate[\psi_{\txtgrandcanonical,\txtvanhowe,I_{L}^{3},\kappa,\Lambda,\sminvtemperature,\smchemicalpotential}](\opfocksegal_{\txtfock}(f) \opfocksegal_{\txtfock}(g)) \\
&=
\onehalf
\rbk{\opreal \opform{q}_{\txtnonzero,\smchemicalpotential}(f,g)
+\imunit \opimag \bkt{f}{g}}
+\opreal \bkt{\mathsf{m}_{\kappa,\Lambda}}{f} \cdot \opreal \bkt{\mathsf{m}_{\kappa,\Lambda}}{g}.
\end{aligned}
\end{equation}

\end{prop}

\begin{proof}
(Preparation): We proceed in the same manner as for the one-point functions.
Here again we argue somewhat formally.
For brevity, let $\physgse$ denote the ground state energy of the van Hove Hamiltonian,
and define the unitary operator $V_{\kappa,\Lambda}
= \opfockweyl_{\txtfock}(\imunit \mathsf{m}_{\kappa,\Lambda})$.

($\opfockcr(f) \opfockan(g)$): Let the operator to be evaluated be $$W
=
\napiernum^{-\sminvtemperature \physham_{\txtvanhowe,I_{L}^{3},\kappa,\Lambda}(\smchemicalpotential)}
\cdot
\opfockcr(f)
\cdot
\opfockan(g).$$
Using Proposition \ref{expedition0012088} for the Hamiltonian
and Fact \ref{expedition0010960} for the creation and annihilation operators, we obtain
\begin{equation}
\begin{aligned}
&\inv{V_{\kappa,\Lambda}}
W
V_{\kappa,\Lambda}
=
\inv{V_{\kappa,\Lambda}}
\napiernum^{-\sminvtemperature \physham_{\txtvanhowe,I_{L}^{3},\kappa,\Lambda}(\smchemicalpotential)}
V_{\kappa,\Lambda}
\cdot
\inv{V_{\kappa,\Lambda}}
\opfockcr(f)
V_{\kappa,\Lambda}
\cdot
\inv{V_{\kappa,\Lambda}}
\opfockan(g)
V_{\kappa,\Lambda}
\\ 
&=
\napiernum^{-\sminvtemperature \physgse}
\napiernum^{-\sminvtemperature \physham_{\txtbsn,\txtfr,I_{L}^{3}}(\smchemicalpotential)}
\rbk{\opfockcr(f) + \frac{\imunit}{\sqrt{2}} \bkt{-\imunit \mathsf{m}_{\kappa,\Lambda}}{f}}
\rbk{\opfockan(g) - \frac{\imunit}{\sqrt{2}} \bkt{g}{-\imunit \mathsf{m}_{\kappa,\Lambda}}}
\\ 
&=
\napiernum^{-\sminvtemperature \physgse}
\napiernum^{-\sminvtemperature \physham_{\txtbsn,\txtfr,I_{L}^{3}}(\smchemicalpotential)}
\rbkleft{\opfockcr(f) \opfockan(g)
-\frac{\imunit}{\sqrt{2}} \opfockcr(f) \bkt{g}{-\imunit \mathsf{m}_{\kappa,\Lambda}}}
\\ 
&\qquad
+\rbkright{\frac{\imunit}{\sqrt{2}} \bkt{-\imunit \mathsf{m}_{\kappa,\Lambda}}{f} \opfockan(g)
+\onehalf \bkt{-\imunit \mathsf{m}_{\kappa,\Lambda}}{f} \bkt{g}{-\imunit \mathsf{m}_{\kappa,\Lambda}}}.
\end{aligned}
\end{equation}
By the cyclicity of the trace, $\sqfun{\trace}{W}
= \sqfun{\trace}{\inv{V_{\kappa,\Lambda}}
W
V_{\kappa,\Lambda}}$ holds.
By the above computation, the evaluation for the ideal Bose gas can ultimately be used.
Using Proposition \ref{expedition0011296} for the one-point functions and the ideal Bose gas,
the one-point functions for the creation and annihilation operators vanish, and we obtain
\begin{equation}
\begin{aligned}
&\oastate[\psi_{\txtgrandcanonical,\txtvanhowe,I_{L}^{3},\kappa,\Lambda,\sminvtemperature,\smchemicalpotential}](\opfockcr(f) \opfockan(g))
=
\napiernum^{\sminvtemperature \physgse}
\frac{1}{\smgrandpartitionfunc_{\txtbsn,\txtfr,I_{L}^{3},\sminvtemperature,\smchemicalpotential}}
\sqfun{\trace}{W} \\
&=
\bkt{g}
{\frac{\napiernum^{-\sminvtemperature (\omega - \smchemicalpotential)}}
{1 - \napiernum^{-\sminvtemperature (\omega - \smchemicalpotential)}}
f}
+\onehalf
\bkt{\mathsf{m}_{\kappa,\Lambda}}{f}
\bkt{g}{\mathsf{m}_{\kappa,\Lambda}}.
\end{aligned}
\end{equation}

($\opfocksegal_{\txtfock}(f) \opfocksegal_{\txtfock}(g)$): The computation proceeds in the same manner as for the pair of creation and annihilation operators.
Let the operator to be evaluated be $$W
=
\napiernum^{-\sminvtemperature \physham_{\txtvanhowe,I_{L}^{3},\kappa,\Lambda}(\smchemicalpotential)}
\cdot
\opfocksegal_{\txtfock}(f)
\cdot
\opfocksegal_{\txtfock}(g).$$
By the result on commutators between the Hamiltonian and the Segal field operator,
\begin{equation}
\begin{aligned}
&\inv{V_{\kappa,\Lambda}}
W
V_{\kappa,\Lambda}
=
\inv{V_{\kappa,\Lambda}}
\napiernum^{-\sminvtemperature \physham_{\txtvanhowe,I_{L}^{3},\kappa,\Lambda}(\smchemicalpotential)}
V_{\kappa,\Lambda}
\cdot
\inv{V_{\kappa,\Lambda}}
\opfocksegal_{\txtfock}(f)
V_{\kappa,\Lambda}
\cdot
\inv{V_{\kappa,\Lambda}}
\opfocksegal_{\txtfock}(g)
V_{\kappa,\Lambda} \\
&=
\napiernum^{-\sminvtemperature \physgse}
\napiernum^{-\sminvtemperature \physham_{\txtbsn,\txtfr,I_{L}^{3}}(\smchemicalpotential)}
\rbk{\opfocksegal_{\txtfock}(f) - \opimag \bkt{-\imunit \mathsf{m}_{\kappa,\Lambda}}{f}}
\rbk{\opfocksegal_{\txtfock}(g) - \opimag \bkt{-\imunit \mathsf{m}_{\kappa,\Lambda}}{g}}
\\ 
&=
\napiernum^{-\sminvtemperature \physgse}
\napiernum^{-\sminvtemperature \physham_{\txtbsn,\txtfr,I_{L}^{3}}(\smchemicalpotential)}
\rbk{\opfocksegal_{\txtfock}(f) - \opreal \bkt{\mathsf{m}_{\kappa,\Lambda}}{f}}
\rbk{\opfocksegal_{\txtfock}(g) - \opreal \bkt{\mathsf{m}_{\kappa,\Lambda}}{g}} \\
\end{aligned}
\end{equation}
By the cyclicity of the trace, $\sqfun{\trace}{W}
= \sqfun{\trace}{\inv{V_{\kappa,\Lambda}}
W
V_{\kappa,\Lambda}}$ holds.
The one-point functions vanish as expectation values for the ideal Bose gas.
Using the evaluation of the two-point functions for the ideal Bose gas, we obtain
\begin{equation}
\begin{aligned}
&\oastate[\psi_{\txtgrandcanonical,\txtvanhowe,I_{L}^{3},\kappa,\Lambda,\sminvtemperature,\smchemicalpotential}]
(\opfocksegal_{\txtfock}(f) \opfocksegal_{\txtfock}(g))
=
\napiernum^{\sminvtemperature \physgse}
\frac{1}{\smgrandpartitionfunc_{\txtbsn,\txtfr,I_{L}^{3},\sminvtemperature,\smchemicalpotential}}
\sqfun{\trace}{W} \\
&=
\onehalf
\rbk{\opreal \opform{q}_{\txtnonzero,\smchemicalpotential}(f,g)
+\imunit \opimag \bkt{f}{g}}
+\opreal \bkt{\mathsf{m}_{\kappa,\Lambda}}{f}
\cdot
\opreal \bkt{\mathsf{m}_{\kappa,\Lambda}}{g}.
\end{aligned}
\end{equation}

\end{proof}

By the above arguments, the grand canonical average of the Weyl operator is obtained.

\begin{prop}\label{expedition0011578}
For any $f
\in \fun{\lp^{2}}{I_{L}^{3}}$,
\begin{equation}
\begin{aligned}
\oastate[\psi_{\txtgrandcanonical,\txtvanhowe,I_{L}^{3},\kappa,\Lambda,\sminvtemperature,\smchemicalpotential}]
(\opfockweyl_{\txtfock}(f))
=
\fnexp{-\imunit
\opreal
\bkt{\mathsf{m}_{\kappa,\Lambda}}{f}
-\frac{1}{4}
\opform{q}_{\txtnonzero,\smchemicalpotential}(f)}.
\end{aligned}
\end{equation}
\end{prop}

\begin{proof}
This can be evaluated by the same method as Proposition \ref{expedition0011299}.
For brevity, let $\physgse$ denote the ground state energy of the van Hove Hamiltonian,
and use the unitary operator $V_{\kappa,\Lambda}
= \opfockweyl_{\txtfock}(\imunit \mathsf{m}_{\kappa,\Lambda})$.

Let the operator to be evaluated from the viewpoint of the trace be $W
=
\opfockweyl_{\txtfock}(f)
\napiernum^{-\sminvtemperature \physham_{\txtvanhowe,I_{L}^{3},\kappa,\Lambda}(\smchemicalpotential)}$.
By Fact \ref{expedition0010960},
\begin{equation}
\begin{aligned}
&\inv{V_{\kappa}}
W
V_{\kappa}
=
\inv{V_{\kappa}}
\opfockweyl_{\txtfock}(f)
\napiernum^{-\sminvtemperature \physham_{\txtvanhowe,I_{L}^{3},\kappa,\Lambda}(\smchemicalpotential)}
V_{\kappa}
\\ 
&=
\napiernum^{-\imunit
\opimag
\bkt{-\imunit \mathsf{m}_{\kappa,\Lambda}}{f}}
\opfockweyl_{\txtfock}(f)
\cdot
\inv{V_{\kappa}}
\napiernum^{-\sminvtemperature \physham_{\txtvanhowe,I_{L}^{3},\kappa,\Lambda}(\smchemicalpotential)}
V_{\kappa}
\\ 
&=
\napiernum^{-\imunit
\opreal
\bkt{\mathsf{m}_{\kappa,\Lambda}}{f}}
\opfockweyl_{\txtfock}(f)
\cdot
\napiernum^{-\sminvtemperature \physgse}
\napiernum^{-\sminvtemperature \physham_{\txtbsn,\txtfr,I_{L}^{3}}(\smchemicalpotential)}
\\ 
&=
\napiernum^{-\imunit
\opreal
\bkt{\mathsf{m}_{\kappa,\Lambda}}{f}}
\napiernum^{-\sminvtemperature \physgse}
\opfockweyl_{\txtfock}(f)
\napiernum^{-\sminvtemperature \physham_{\txtbsn,\txtfr,I_{L}^{3}}(\smchemicalpotential)}.
\end{aligned}
\end{equation}
By the cyclicity of the trace,
\begin{equation}
\begin{aligned}
&\sqfun{\trace}
{W}
=
\sqfun{\trace}
{\inv{V_{\kappa,\Lambda}}
W
V_{\kappa,\Lambda}}
\\ 
&=
\napiernum^{-\imunit
\opreal \bkt{\mathsf{m}_{\kappa,\Lambda}}{f}}
\napiernum^{-\sminvtemperature \physgse}
\sqfun{\trace}
{\opfockweyl_{\txtfock}(f)
\napiernum^{-\sminvtemperature \physham_{\txtbsn,\txtfr,I_{L}^{3}}(\smchemicalpotential)}}
\\ 
&=
\napiernum^{-\imunit
\opreal
\bkt{\mathsf{m}_{\kappa,\Lambda}}{f}}
\smgrandpartitionfunc_{\txtvanhowe,I_{L}^{3},\kappa,\Lambda,\sminvtemperature,\smchemicalpotential}
\cdot
\psi_{\txtgrandcanonical,\sminvtemperature,\smchemicalpotential}(\opfockweyl_{\txtfock}(f))
\end{aligned}
\end{equation}
Summarizing the above arguments,
\begin{equation}
\begin{aligned}
&\oastate[\psi_{\txtgrandcanonical,\txtvanhowe,I_{L}^{3},\kappa,\Lambda,\sminvtemperature,\smchemicalpotential}](\opfockweyl_{\txtfock}(f))
\\ 
&=
\napiernum^{-\imunit
\opreal
\bkt{\mathsf{m}_{\kappa,\Lambda}}{f}}
\fnexp{-\frac{1}{4}
\opform{q}_{\txtnonzero,\smchemicalpotential}(f)}.
\end{aligned}
\end{equation}
\end{proof}

We also prepare the two-point function of the Weyl operator.

\begin{prop}\label{expedition0011592}
For any $f,g
\in \sphilb{H}$,
\begin{equation}
\begin{aligned}
&\oastate[\psi_{\txtgrandcanonical,\txtvanhowe,I_{L}^{3},\kappa,\Lambda,\sminvtemperature,\smchemicalpotential}]
(\opfockweyl_{\txtfock}(f)
\opfockweyl_{\txtfock}(g))
\\ 
&=
\fnexp{-\frac{\imunit}{2}
\opimag
\bkt{f}{g}
-\imunit
\opreal
\bkt{\mathsf{m}_{\kappa,\Lambda}}{f+g}
-\frac{1}{4}
\opform{q}_{\txtnonzero,\smchemicalpotential}(f+g)}.
\end{aligned}
\end{equation}
\end{prop}

\begin{proof}
By the Weyl relations, $$\opfockweyl_{\txtfock}(f)
\opfockweyl_{\txtfock}(g)
=
\napiernum^{-\frac{\imunit}{2}
\opimag
\bkt{f}{g}}
\opfockweyl_{\txtfock}(f+g)$$holds.
Then by Proposition \ref{expedition0011578},
\begin{equation}
\begin{aligned}
&\oastate[\psi_{\txtgrandcanonical,\txtvanhowe,I_{L}^{3},\kappa,\Lambda,\sminvtemperature,\smchemicalpotential}]
(\opfockweyl_{\txtfock}(f)
\opfockweyl_{\txtfock}(g))
\\ 
&=
\napiernum^{-\frac{\imunit}{2}
\opimag
\bkt{f}{g}}
\oastate[\psi_{\txtgrandcanonical,\txtvanhowe,I_{L}^{3},\kappa,\Lambda,\sminvtemperature,\smchemicalpotential}]
(\opfockweyl_{\txtfock}(f+g))
\\ 
&=
\napiernum^{-\frac{\imunit}{2}
\opimag
\bkt{f}{g}}
\napiernum^{-\imunit
\opreal
\bkt{\mathsf{m}_{\kappa,\Lambda}}{f+g}}
\napiernum^{-\frac{1}{4}
\opform{q}_{\txtnonzero,\smchemicalpotential}(f+g)}.
\end{aligned}
\end{equation}
\end{proof}

By the following proposition, the grand canonical state for the van Hove model is clearly separated into the term requiring infrared and ultraviolet divergence treatment and the grand canonical state for the free field.

\begin{prop}\label{expedition0011594}
Let $\oastate[\psi_{\txtgrandcanonical,\txtfr,I_{L}^{3},\sminvtemperature,\smchemicalpotential}]$ be the grand canonical state for the free field.
For any $f
\in \sphilb{H}$, the grand canonical state $\oastate[\psi_{\txtgrandcanonical,\txtvanhowe,I_{L}^{3},\kappa,\Lambda,\sminvtemperature,\smchemicalpotential}]$ for the van Hove model satisfies
\begin{equation}
\begin{aligned}
&\oastate[\psi_{\txtgrandcanonical,\txtvanhowe,I_{L}^{3},\kappa,\Lambda,\sminvtemperature,\smchemicalpotential}]
(\opfockweyl_{\txtfock}(f))
=
\napiernum^{-\imunit
\opreal
\bkt{\mathsf{m}_{\kappa,\Lambda}}{f}}
\fun{\oastate[\psi_{\txtgrandcanonical,\txtfr,I_{L}^{3},\sminvtemperature,\smchemicalpotential}]}
{\opfockweyl_{\txtfock}(f)}.
\end{aligned}
\end{equation}
\end{prop}

\begin{proof}
This follows from Proposition \ref{expedition0011578} and the result for the free Bose gas.
\end{proof}

Let us compute the action on the Weyl operator for the automorphism group defined in Section \ref{expedition0011278}.

\begin{prop}\label{expedition0011094}
For any $t
\in \fldreal$, defining the $\ast$-endomorphism $\alpha_{\txtvanhowe,I_{L}^{3},\kappa,\Lambda,t}$ of the abstract Weyl algebra by $$\alpha_{\txtvanhowe,I_{L}^{3},\kappa,\Lambda,t}(\opfockweyl_{\txtfock}(f))
=
\napiernum^{\imunit \mathsf{M}_{\kappa,\Lambda,t}(f)}
\opfockweyl_{\txtfock}(\napiernum^{\imunit t \omega} f),$$
this is indeed the action of the Hamiltonian of the van Hove model on the Weyl operator.
\end{prop}

\begin{rem}
Although we have only computed the action for the bounded system with cutoffs here,
the action can be computed by exactly the same calculation for the system without cutoffs and for infinite systems,
and the results coincide.
In the subsequent discussion, we also use the above proposition for the system without cutoffs and for infinite systems.
\end{rem}

\begin{proof}
($\ast$-isomorphism: preservation of the Weyl relations): First, $\alpha_{\txtvanhowe,I_{L}^{3},\kappa,\Lambda,0}
= \idone$ is clear.
Since the inverse map for each $t$ is the map with the parameter replaced by $-t$,
it suffices to show that it is a homomorphism.

The operator $\napiernum^{\imunit t \omega}$ is unitary, so it preserves the symplectic form defined from the inner product.
Since $\mathsf{M}_{\kappa,\Lambda,t}$ arises from the inner product, it is additive.
Then by direct computation,
\begin{equation}
\begin{aligned}
&\alpha_{\txtvanhowe,I_{L}^{3},\kappa,\Lambda,t}(\opfockweyl_{\txtfock}(f))
\cdot
\alpha_{\txtvanhowe,I_{L}^{3},\kappa,\Lambda,t}(\opfockweyl_{\txtfock}(g)) \\
&=
\napiernum^{\imunit \rbk{\mathsf{M}_{\kappa,\Lambda,t}(f) + \mathsf{M}_{\kappa,\Lambda,t}(g)}}
\opfockweyl_{\txtfock}(\napiernum^{\imunit t \omega} f)
\opfockweyl_{\txtfock}(\napiernum^{\imunit t \omega} g) \\
&=
\napiernum^{\imunit \rbk{\mathsf{M}_{\kappa,\Lambda,t}(f) + \mathsf{M}_{\kappa,\Lambda,t}(g)}}
\napiernum^{-\frac{\imunit}{2} \opimag \bkt{\napiernum^{\imunit t \omega} f}{\napiernum^{\imunit t \omega} g}}
\opfockweyl_{\txtfock}(\napiernum^{\imunit t \omega}(f+g)), \\
&=
\napiernum^{-\frac{\imunit}{2} \opimag \bkt{f}{g}}
\napiernum^{\imunit \mathsf{M}_{\kappa,\Lambda,t}(f+g)}
\opfockweyl_{\txtfock}(\napiernum^{\imunit t \omega}(f+g)), \\
&\alpha_{\txtvanhowe,I_{L}^{3},\kappa,\Lambda,t}(\opfockweyl_{\txtfock}(f) \opfockweyl_{\txtfock}(g)) \\
&=
\napiernum^{-\frac{\imunit}{2} \opimag \bkt{f}{g}}
\alpha_{\txtvanhowe,I_{L}^{3},\kappa,\Lambda,t}(\opfockweyl_{\txtfock}(f+g)) \\
&=
\napiernum^{-\frac{\imunit}{2} \opimag \bkt{f}{g}}
\napiernum^{\imunit \mathsf{M}_{\kappa,\Lambda,t}(f+g)}
\opfockweyl_{\txtfock}(\napiernum^{\imunit t \omega}(f+g)),
\end{aligned}
\end{equation}
and these are equal.

(Adjoint): Since the map $\mathsf{M}_{\kappa,\Lambda,t}$ satisfies $\mathsf{M}_{\kappa,\Lambda,t}(-f)
= - \mathsf{M}_{\kappa,\Lambda,t}(f)$,
\begin{equation}
\begin{aligned}
&\faadj{\alpha_{\txtvanhowe,I_{L}^{3},\kappa,\Lambda,t}(\opfockweyl_{\txtfock}(f))}
=
\napiernum^{-\imunit \mathsf{M}_{\kappa,\Lambda,t}(f)}
\opfockweyl_{\txtfock}(-\napiernum^{\imunit t \omega} f) \\
&=
\napiernum^{\imunit \mathsf{M}_{\kappa,\Lambda,t}(-f)}
\opfockweyl_{\txtfock}(\napiernum^{\imunit t \omega} (-f)) \\
&=
\fun{\alpha_{\txtvanhowe,I_{L}^{3},\kappa,\Lambda,t}}{\opfockweyl_{\txtfock}(-f)}
=
\fun{\alpha_{\txtvanhowe,I_{L}^{3},\kappa,\Lambda,t}}{\faadj{\opfockweyl_{\txtfock}(f)}}
\end{aligned}
\end{equation}
holds.

(Group property): First,
\begin{equation}
\begin{aligned}
\alpha_{\txtvanhowe,I_{L}^{3},\kappa,\Lambda,t}(\opfockweyl_{\txtfock}(f))
&=
\napiernum^{\imunit \mathsf{M}_{\kappa,\Lambda,t}(f)} \opfockweyl_{\txtfock}(\napiernum^{\imunit t \omega} f), \\
\alpha_{\txtvanhowe,I_{L}^{3},\kappa,\Lambda,s}(\opfockweyl_{\txtfock}(f))
&=
\napiernum^{\imunit \mathsf{M}_{\kappa,\Lambda,s}(f)} \opfockweyl_{\txtfock}(\napiernum^{\imunit s \omega} f).
\end{aligned}
\end{equation}
Then $$\alpha_{\txtvanhowe,I_{L}^{3},\kappa,\Lambda,t}
\circ
\alpha_{\txtvanhowe,I_{L}^{3},\kappa,\Lambda,s}(\opfockweyl_{\txtfock}(f))
=
\napiernum^{\imunit \mathsf{M}_{\kappa,\Lambda,s}(f)}
\napiernum^{\imunit \mathsf{M}_{\kappa,\Lambda,t}(\napiernum^{\imunit s \omega} f)}
\opfockweyl_{\txtfock}(\napiernum^{\imunit t \omega} \napiernum^{\imunit s \omega} f)$$
holds.
On the other hand, $$\alpha_{\kappa,\Lambda,t+s}(\opfockweyl_{\txtfock}(f))
=
\napiernum^{\imunit \mathsf{M}_{\kappa,\Lambda,t+s}(f)}
\opfockweyl_{\txtfock}(\napiernum^{\imunit (t+s) \omega} f)$$holds.
This is the cocycle condition, which was proved in Proposition \ref{expedition0011238}.
\end{proof}

We also prepare the expectation value of the resolvent for the discussion on the resolvent algebra.

\begin{prop}\label{expedition0011613}
The grand canonical average of the resolvent $\oaresolvent(\lambda,f)
= \opfnresolvent{\imunit \lambda - \opfocksegal_{\txtfock}(f)}$ is
\begin{equation}
\begin{aligned}
&\fun{\oastate[\psi_{\txtgrandcanonical,\txtvanhowe,I_{L}^{3},\sminvtemperature,\smchemicalpotential}]}
{\oaresolvent(\lambda,f)}
\\ 
&=
\imunit
\int_0^{(\sgn \lambda) \infty}
\fnexp{-\rbk{\lambda - \imunit \opreal \mathsf{m}_{\kappa,\Lambda}(f)} t
-\frac{t^2}{4} \opform{q}_{\txtnonzero,\smchemicalpotential}(f)}
\opdmsr{t}.
\end{aligned}
\end{equation}

\end{prop}

\begin{proof}
By the Laplace transform, or by the argument in \cite{BuchholzGrundling2}, the resolvent is expressed in terms of the Weyl operator as
\begin{equation}
\begin{aligned}
&\opfnresolvent{\imunit \lambda - \opfocksegal_{\txtfock}(f)}
=
\imunit
\int_0^{(\sgn \lambda) \infty}
\napiernum^{\imunit t(\imunit \lambda - \opfocksegal_{\txtfock}(f))}
\opdmsr{t}
\\ 
&=
\imunit
\int_0^{(\sgn \lambda) \infty}
\napiernum^{-\lambda t}
\napiernum^{-\imunit t \opfocksegal_{\txtfock}(f)}
\opdmsr{t}
=
\imunit
\int_0^{(\sgn \lambda) \infty}
\napiernum^{-\lambda t}
\opfockweyl_{\txtfock}(-tf)
\opdmsr{t}.
\end{aligned}
\end{equation}
Then using Proposition \ref{expedition0011578},
\begin{equation}
\begin{aligned}
&\fun{\oastate[\psi_{\txtgrandcanonical,\txtvanhowe,I_{L}^{3},\sminvtemperature,\smchemicalpotential}]}
{\oaresolvent(\lambda,f)}
=
\fun{\oastate[\psi_{\txtgrandcanonical,\txtvanhowe,I_{L}^{3},\sminvtemperature,\smchemicalpotential}]}
{\opfnresolvent{\imunit \lambda - \opfocksegal_{\txtfock}(f)}}
\\ 
&=
\imunit
\int_0^{(\sgn \lambda) \infty}
\napiernum^{-\lambda t}
\fun{\oastate[\psi_{\txtgrandcanonical,\txtvanhowe,I_{L}^{3},\sminvtemperature,\smchemicalpotential}]}
{\opfockweyl_{\txtfock}(-tf)}
\opdmsr{t}
\\ 
&=
\imunit
\int_0^{(\sgn \lambda) \infty}
\napiernum^{-\lambda t}
\napiernum^{-\imunit
\opreal
\bkt{\mathsf{m}_{\kappa,\Lambda}}{-tf}}
\fnexp{-\frac{t^2}{4}
\bkt{f}
{K_{\sminvtemperature,\smchemicalpotential}
f}}
\opdmsr{t}
\\ 
&=
\imunit
\int_0^{(\sgn \lambda) \infty}
\fnexp{-\rbk{\lambda - \imunit \opreal \mathsf{m}_{\kappa,\Lambda}(f)} t
-\frac{t^2}{4} \opform{q}_{\txtnonzero,\smchemicalpotential}(f)}
\opdmsr{t}.
\end{aligned}
\end{equation}

\end{proof}

Let us also prepare the two-point function of the resolvent.

\begin{prop}\label{expedition0011614}
For the resolvent $\oaresolvent(\lambda,f)
= \opfnresolvent{\imunit \lambda - \opfocksegal_{\txtfock}(f)}$,
\begin{equation}
\begin{aligned}
&\fun{\oastate[\psi_{\txtgrandcanonical,\txtvanhowe,I_{L}^{3},\sminvtemperature,\smchemicalpotential}]}
{\oaresolvent(\lambda,f)
\oaresolvent(\mu,g)}
\\ 
&=
\imunit
\int_0^{(\sgn \lambda) \infty}
\int_0^{(\sgn \mu) \infty}
\napiernum^{-\lambda s - \mu t}
\mathsf{T}_{\sminvtemperature,\smchemicalpotential,\kappa,\Lambda}(s,f;t,g)
\opdmsr{s}
\opdmsr{t},
\\ 
&\mathsf{T}_{\sminvtemperature,\smchemicalpotential,\kappa,\Lambda}(s,f;t,g)
\\
&=
\fnexp{-\frac{1}{4}
\opform{q}_{\txtnonzero,\smchemicalpotential}(sf+tg)
+\imunit \opreal \mathsf{m}_{\kappa,\Lambda}(sf+tg)
-\frac{\imunit}{2} st \opimag \bkt{f}{g}}.
\end{aligned}
\end{equation}

\end{prop}

\begin{proof}
It suffices to combine the product into one using the Weyl relations and then apply Proposition \ref{expedition0011613}.
\end{proof}

\section{Infinite systems on the Weyl algebra}\label{infinite-systems-on-the-weyl-algebra}

By Proposition \ref{expedition0011594}, BEC can essentially be treated by the free-field argument, so it suffices to consider only the situation \(\smchemicalpotential
= 0\) where BEC can occur.

\subsection{Evaluation of expectation values}\label{evaluation-of-expectation-values}

We first discuss with infrared and ultraviolet cutoffs, and then discuss the situation below the critical temperature under the removal of infrared and ultraviolet cutoffs.

\begin{thm}\label{expedition0011612}
The expectation value of the Weyl operator in the $\sminvtemperature$-KMS state $\oastate[\psi_{\txtvanhowe,\sminvtemperature,\kappa,\Lambda}]$ for the van Hove model with infrared and ultraviolet cutoffs can be written as $$\oastate[\psi_{\txtvanhowe,\sminvtemperature,\kappa,\Lambda}]
(\opfockweyl_{\txtfock}(f))
=
\fnexp{-\imunit
\opreal
\bkt{\mathsf{m}_{\kappa,\Lambda}}{f}
-\oneoverfour
\opform{q}_{\txtnonzero}(f)
-\oneoverfour
\opform{q}_{0}(f)}.$$
\end{thm}

\begin{rem}
The difference from the free Bose gas is precisely the term $\napiernum^{-\imunit
\opreal
\bkt{\mathsf{m}_{\kappa,\Lambda}}{f}}$ arising from the addition of the field interaction.
\end{rem}

\begin{proof}
It suffices to consider the infinite volume limit discussed in \cite{YoshitsuguSekine004,AsaoArai28} applied to Proposition \ref{expedition0011594}.
\end{proof}

We next remove the infrared and ultraviolet cutoffs. We make one remark concerning the procedure for taking limits. Although we have been working on Fock space up to this point, from the viewpoint of representation theory we should regard the discussion as taking place in the Fock representation of the abstract Weyl algebra \(\oaweyl(\sphilb{D}_{\txtirsingular,\sminvtemperature})\). In the discussion of removing the infrared and ultraviolet cutoffs as well, the state \(\oastate[\psi_{\txtvanhowe,\sminvtemperature,\kappa,\Lambda}]\) should be regarded as the state \(\oastate[\tilde{\psi}_{\txtvanhowe,\sminvtemperature,\kappa,\Lambda}]\) on \(\oaweyl(\sphilb{D}_{\txtirsingular,\sminvtemperature})\) composed with the Fock representation \(\oarepn_{\txtfock}\), that is, \(\oastate[\tilde{\psi}_{\txtvanhowe,\sminvtemperature,\kappa,\Lambda}]
\circ \oarepn_{\txtfock}\), and the limit should be taken for \(\oastate[\tilde{\psi}_{\txtvanhowe,\sminvtemperature,\kappa,\Lambda}]\) in the state space on the algebra \(\oaweyl(\sphilb{D}_{\txtirsingular,\sminvtemperature})\). For brevity of notation, we write \(\oastate[\tilde{\psi}_{\txtvanhowe,\sminvtemperature,\kappa,\Lambda}]\) simply as \(\oastate[\psi_{\txtvanhowe,\sminvtemperature,\kappa,\Lambda}]\).

\begin{thm}\label{expedition0011637}
\begin{enumerate}
\item
Convergence of the automorphism group: For any $A
\in \oaweyl_{\txtirsingular,\sminvtemperature}$ and $t
\in \fldreal$, the limit $$\alpha_{\txtvanhowe,t}(A)
=
\lim_{\kappa \to 0, \Lambda \to \infty}
\alpha_{\txtvanhowe,t,\kappa,\Lambda}(A)$$exists in the norm topology,
and the family $\fml{\alpha_{\txtvanhowe,t}}{t \in \fldreal}$ gives a one-parameter automorphism group on $\oaweyl_{\txtirsingular,\sminvtemperature}$.

\item
Convergence of the KMS state: For any $A
\in \oaweyl_{\txtirsingular,\sminvtemperature}$, there exists a state $\oastate[\psi_{\txtvanhowe,\sminvtemperature}]$ on $\oaweyl_{\txtirsingular,\sminvtemperature}$ satisfying $$\oastate[\psi_{\txtvanhowe,\sminvtemperature}](A)
=
\lim_{\kappa \to 0,\Lambda \to \infty}
\oastate[\psi_{\txtvanhowe,\sminvtemperature,\kappa,\Lambda}](A).$$
In particular, for the Weyl operator, $$\oastate[\psi_{\txtvanhowe,\sminvtemperature}]
(\opfockweyl(f))
=
\lim_{\kappa \to 0,\Lambda \to \infty}
\oastate[\psi_{\txtvanhowe,\sminvtemperature,\kappa,\Lambda}]
(\opfockweyl(f)),
\quad
f \in \sphilb{D}_{\txtirsingular,\sminvtemperature}$$holds,
and the right-hand side is obtained as the natural limit of the expression in Theorem \ref{expedition0011612}.
In particular, the expectation value for the Weyl operator is $$\oastate[\psi_{\txtvanhowe,\sminvtemperature}]
(\opfockweyl(f))
=
\fnexp{-\imunit
\opreal
\bkt{\mathsf{m}}{f}
-\oneoverfour
\opform{q}_{\txtnonzero}(f)
-\oneoverfour
\opform{q}_{0}(f)}.$$

\item
The state $\oastate[\psi_{\txtvanhowe,\sminvtemperature}]$ is a $\sminvtemperature$-KMS state for the automorphism group $\alpha_{\txtvanhowe}$.
\end{enumerate}
\end{thm}

\begin{proof}
For the removal limit of the infrared and ultraviolet cutoffs, it suffices to consider the objects obtained by simply taking the limits in $\kappa,\Lambda$.
Therefore it is sufficient to discuss whether these objects converge in the appropriate limits.

(Convergence of the automorphism group): By Proposition \ref{expedition0011094} and the remark immediately following it,
the limit automorphism is defined for each $t$.
The removal limit concerns only $\mathsf{m}_{\kappa,\Lambda}$,
and the convergence in norm is clear.

(Convergence of the KMS state): The cutoffs act only on the mean functional.
It remains to discuss convergence on $\opfockweyl(f)$ for $f
\in \sphilb{D}_{\txtirsingular,\sminvtemperature}$,
which is clear.
The KMS property of the limit holds because the defining equations from the bounded system are algebraic
and are preserved in the limit.
\end{proof}

\subsection{Selection criterion for physical phonons and the no-go theorem for BEC}\label{expedition0012085}

Toward the following theorem, we define the Segal field operator satisfying the self-consistency condition of \cite{VIYukalov001} as \(\opfocksegal_{\txtselfconsistent}(f)
=
\opfocksegal(f)
+\opreal \mathsf{m}(f)\).

\begin{thm}[Selection criterion for physical phonons]\label{expedition0012077}
For the $\sminvtemperature$-KMS state $\oastate[\psi_{\txtvanhowe,\sminvtemperature}]$ of the van Hove model, $\fun{\oastate[\psi_{\txtvanhowe,\sminvtemperature}]}
{\opfocksegal_{\txtselfconsistent}(f)}
= 0$ holds.
\end{thm}

\begin{proof}
Differentiating $$\oastate[\psi_{\txtvanhowe,\sminvtemperature}]
(\opfockweyl(tf))
=
\fnexp{-\imunit t
\opreal
\bkt{\mathsf{m}}{f}
-\oneoverfour t^2
\opform{q}_{\txtnonzero}(f)
-\oneoverfour t^2
\opform{q}_{0}(f)}$$from Theorem \ref{expedition0011637} with respect to $t$ and setting $t
= 0$, we obtain $$\oastate[\psi_{\txtvanhowe,\sminvtemperature}]
(\opfocksegal(f))
=
-\fnrestr{\imunit \opod{t}
\oastate[\psi_{\txtvanhowe,\sminvtemperature}]
(\opfockweyl(tf))}
{t = 0}
=
-\opreal
\bkt{\mathsf{m}}{f}.$$
Rearranging yields the desired result.
\end{proof}

\begin{rem}[Self-consistency condition of \cite{VIYukalov001} and nonstandard nature of the field operator]\label{expedition0012090}
The self-consistency condition imposed in \cite{VIYukalov001} requires $\physmean{\opfocksegal_{\txtselfconsistent}(f)}
= 0$ for the phonon field operator $\opfocksegal_{\txtselfconsistent}(f)$.
This can be interpreted as a condition for defining phonons as small-amplitude fluctuations from the background field,
namely a physical selection principle corresponding to the absorption of macroscopic components into the background field.
Since phonons as quasiparticles are by definition constructed so as not to carry a mean-field component,
the so-called phonon condensation is an object that should be excluded by redefinition of the field.

However, the field operator defined to satisfy this condition does not coincide,
at least in the van Hove model, with the field given by the Bogoliubov-type transformation $\opfocksegal(f)
\mapsto \opfocksegal(f) - \opimag \mathsf{m}(f)$ naturally obtained from automorphisms of the canonical commutation relations,
that is, from the conjugation action by Weyl operators.
The phonons in \cite{VIYukalov001} are defined not simply as the Bogoliubov transform of the original field,
but as the fluctuations remaining after removing the lattice distortion as the mean displacement of atoms as a background field.
This operation is generally not realized by automorphisms of the canonical commutation relations,
and the field operator $\opfocksegal_{\txtselfconsistent}(f)$ should be regarded as a field redefined based on physical requirements.
From this viewpoint,
Theorem \ref{expedition0012077} can be interpreted not as merely reproducing the self-consistency condition,
but as a statement that extracts the underlying selection criterion for physical phonons from the perspective of operator algebras.

However, the self-consistency condition is merely a constraint concerning the definition of the field,
and is not a condition that prescribes properties of states,
in particular requirements for equilibrium states.
While the redefinition in the sense of eliminating the expectation value of the field is achieved by introducing the operator $\opfocksegal_{\txtselfconsistent}(f)$,
this alone is insufficient to establish the no-go theorem for BEC.
Therefore, in order to mathematically establish the no-go theorem for BEC,
an additional selection principle for states is essentially required.

In this sense, the standpoint of this paper does not conflict with the physical claims of \cite{VIYukalov001}.
Rather, it should be clearly formulated by decomposing them into two independent elements,
namely the redefinition of the field and the selection principle for equilibrium states.
\end{rem}

In light of the above remark, we consider the state selection principle as follows. We regard the equilibrium state for phonons in \cite{VIYukalov001} as the assertion that, after absorbing various situations such as lattice distortions that can arise from infrared divergences, the state should be a mild state reformulated as a new equilibrium state. Here, as a selection principle for states possessing mildness in which macroscopic temporal anomalies have disappeared after a sufficiently long time, we adopt the time cluster property.

\begin{defn}[State selection principle for physical phonons]\label{expedition0012082}
As the selection principle for equilibrium states of phonons, we specify the preservation of mixing,
in particular the time cluster property \cite{BratteliRobinson1,BratteliRobinson2}.
\end{defn}

Under this definition, the following theorem is obtained.

\begin{thm}[No-go theorem for BEC of the van Hove model or phonons]\label{expedition0012089}
Under the state selection principle of Definition \ref{expedition0012082}, the van Hove model does not exhibit BEC.
\end{thm}

\begin{proof}
By \cite{AsaoArai26,YoshitsuguSekine004}, for the free Bose gas, the validity of the time cluster property and the no-go theorem for BEC are equivalent.
Since the BEC behavior of the van Hove model is determined in the same way as for the free Bose gas,
the van Hove model does not exhibit BEC under the assumption of the time cluster property.
\end{proof}

\begin{rem}
Also by \cite{AsaoArai26,YoshitsuguSekine004}, the validity of the time cluster property,
the validity of the space cluster property,
and the condition on the order parameter on the resolvent algebra defined in those references are all equivalent.

If the uniform component of the field is to be absorbed into the background (distortion),
then the component remaining at long distances should not be regarded as a phonon.
At least in the van Hove model,
the very definition of phonons can be regarded as prohibiting the occurrence of off-diagonal long-range order.
Mathematically, this is the validity of the space cluster property.
At least in the van Hove model,
these cluster property assumptions are equivalent as conditions prohibiting the occurrence of BEC.
\end{rem}

By the above arguments, the contribution corresponding to the zero-mode condensation term \(\smnumberdensity_0(\sminvtemperature)\) appearing in the quasi-free state is physically not permitted as long as phonons are defined as linear fluctuations in equilibrium, and only the sector \cite{IzumiOjima002} satisfying \(\smnumberdensity_0(\sminvtemperature)
= 0\) is appropriate as a representation for phonons.

\subsection{Reduction of the algebra of physical observables for higher-order nonlinear dispersion and the no-go theorem}\label{expedition0012081}

The results of the preceding section hold regardless of the degree of the dispersion relation. However, as mentioned in Section \ref{expedition0012078}, there is the problem of treating the background field. The following theorem can be proved regarding the circumstances under which higher-order nonlinear dispersion is absorbed into the background.

\begin{thm}\label{expedition0011638}
Let the dispersion relation of the van Hove model be $\omega_s(k)
= \abs{k}^{s}$ for $s
> 2$.
Then $\opform{q}_0(f)
= 0$ holds for every $f
\in \sphilb{D}_{\txtirsingular,\sminvtemperature}$ that generates the algebra of physical observables.
In particular, Bose-Einstein condensation vanishes on the algebra of physical observables incorporating the infrared singularity condition.
\end{thm}

\begin{proof}
In this case, $f
\in \dom \mathsf{m}$ must satisfy $$\abs{\int_{k \in \greuctr{k}{3}}
\frac{\faftr{f}}{\abs{k}^{s}}
\opdmsr{k}}
\sim
\int_{k \in \greuctr{k}{3}}
\abs{k}^{2-s} \abs{\faftr{f}}
\opdmsr{k}
<
\infty.$$
In particular, for any $\delta
> 0$, the singularity around the origin must satisfy $\faftr{f}(k)
= \fun{O}{\abs{k}^{s-2}}$,
and in particular for $s
> 2$, $\faftr{f}(0)
= 0$ must hold.
Since $\opform{q}_0(f)
= 0$ holds for such $f$,
Bose-Einstein condensation must vanish on the algebra of physical observables.
\end{proof}

\begin{rem}
For the higher-order dispersion relation $\omega_s$ with any $s > 2$,
the occurrence of BEC is eliminated on the algebra of physical observables incorporating the infrared singularity condition.
However, this does not hold for $1 \leq s \leq 2$.
In particular, the no-go theorem for BEC with the linear dispersion, which is the original acoustic phonon,
requires a state selection principle, at least as an argument for the van Hove model, in the form of Theorem \ref{expedition0012089}.

However, since the model under discussion is a non-physical model prioritizing mathematical tractability,
this may simply be a deficiency of the model itself.
\end{rem}

\subsection{Physical requirements for equilibrium states and the no-go theorem for quasiparticle BEC}\label{physical-requirements-for-equilibrium-states-and-the-no-go-theorem-for-quasiparticle-bec}

Here we clarify the position of the state selection principle for physical phonons introduced in the preceding sections. In \cite{VIYukalov001}, the important point is the self-consistent construction of the Hamiltonian, and we emphasize again at the outset that the discussion in this paper, where the Hamiltonian is fixed to the van Hove model, is fundamentally different in nature.

As confirmed in Section \ref{expedition0012081}, if the dispersion relation is of sufficiently high order, infrared divergences are mathematically treated as a reduction of the algebra of physical observables, and BEC vanishes as a result. Therefore, the essential problem lies in how to exclude quasiparticle BEC for low-order dispersion, in particular for the linear dispersion \(s
= 1\).

In the Hubbard--phonon interaction system of \cite{YoshitsuguSekine001,YoshitsuguSekine002}, and in the van Hove model which behaves identically for the boson field, the Hamiltonian is equivalent to the free Hamiltonian via a unitary transformation under the infrared regularization condition. The problem is therefore not the construction of the Hamiltonian, but rather, regarding the no-go theorem for BEC in equilibrium states in \cite{VIYukalov001}, the question is rephrased as what kind of equilibrium states should be regarded as equilibrium states for phonons. What we actually adopted was Definition \ref{expedition0012082}. In particular, the equilibrium state for phonons was regarded as a mild state in which correlations stabilize in the infinite time limit and which does not possess anomalous long-time memory or macroscopic order.

This requirement was formulated as the time cluster property. By known results for the free Bose gas \cite{AsaoArai26,YoshitsuguSekine004}, the time cluster property is equivalent to the space cluster property, and furthermore these are equivalent to the no-go theorem for Bose-Einstein condensation. The space cluster property is a property asserting spatial mildness in addition to time mildness, and this is consistent with the physical requirement representing the absorption of distortions into the background; we also make use of the fact that these coincide for the van Hove model and the free Bose gas.

Since the van Hove model is formally equivalent to the free Bose gas via an appropriate transformation, this equivalence carries over directly to the model in question. Therefore, quasiparticle BEC does not occur in equilibrium states satisfying time and space cluster properties. In this sense, the no-go theorem for quasiparticle BEC in this paper proposes an understanding not as a result depending on the details of a particular Hamiltonian, but as a consequence of the physical requirements imposed on equilibrium states.

This standpoint is also consistent with the self-consistent definition of phonons in \cite{VIYukalov001}. Namely, since phonons are defined as small-amplitude fluctuations from the background, their macroscopic component is absorbed into the background and does not remain as a phonon degree of freedom. In this paper, we translate this structure as a condition on the side of states, expressing the mildness of the equilibrium state through cluster properties. In particular, the no-go theorem for BEC of quasiparticles with the original linear dispersion has been formulated in a form equivalent to the selection of the class of physically admissible equilibrium states.

\section{Finite temperature and equilibrium states on the resolvent algebra}\label{expedition0012002}

The resolvent algebra was defined in Section \ref{expedition0012083}. Here we begin with the definition of the van Hove model on the resolvent algebra. We then appropriately transplant the bounded-system, finite-temperature setup from Section \ref{expedition0011587}, where bounded systems and finite temperature on the Weyl algebra were discussed, and in particular investigate the ideal structure of the resolvent algebra following the approach of \cite{DetlevBuchholz001,YoshitsuguSekine004,YoshitsuguSekine005}.

\subsection{Formulation of the van Hove model and basic arguments}\label{formulation-of-the-van-hove-model-and-basic-arguments}

We define each object based on the discussion on the Weyl algebra.

For any \(t
\in \fldreal\), we define the \(\ast\)-endomorphism of the abstract resolvent algebra by \begin{equation}
\begin{aligned}
\alpha_{1,\kappa,\Lambda,t}(\oaresolvent(z,f))
&=
\fun{\oaresolvent}{z + \imunit \mathsf{M}_{\kappa,\Lambda,t}(f), f}, \\
\alpha_{2,t}(\oaresolvent(z,f))
&=
\fun{\oaresolvent}{z, \napiernum^{\imunit t \omega} f}, \\
\alpha_{\kappa,\Lambda,t}(\oaresolvent(z,f))
&=
\funrbk{\alpha_{1,\kappa,\Lambda,t} \circ \alpha_{2,t}}{\oaresolvent(z,f)} \\
&=
\fun{\oaresolvent}{z + \imunit \mathsf{M}_{\kappa,\Lambda,t}(f), \napiernum^{\imunit t \omega} f}
\end{aligned}
\end{equation} and call it the automorphism group of the van Hove model with infrared and ultraviolet cutoffs, or simply the automorphism group of the van Hove model with cutoffs or the automorphism group of the van Hove model.

The maps \(\alpha_{t}\) and \(\alpha_{1,t}\) without cutoffs are defined as \(\ast\)-endomorphisms of the abstract resolvent algebra \(\oaresolventalgebra(\dom \mathsf{m}, \sigma)\), in a form arising from the domain of the functional \(\mathsf{m}\): \begin{equation}
\begin{aligned}
\alpha_{1,t}(\oaresolvent(z,f))
&=
\fun{\oaresolvent}{z + \imunit \mathsf{M}_{t}(f), f}, \\
\alpha_{t}(\oaresolvent(z,f))
&=
\funrbk{\alpha_{1,t} \circ \alpha_{2,t}}{\oaresolvent(z,f)} \\
&=
\fun{\oaresolvent}{z + \imunit \mathsf{M}_{t}(f), \napiernum^{\imunit t \omega} f}
\end{aligned}
\end{equation} and call it the automorphism group of the van Hove model, or the automorphism group of the van Hove model without cutoffs.

\begin{prop}\label{expedition0012084}
The above $\alpha_{\kappa,\Lambda}$ is an automorphism group.
In particular, even when the domain is restricted to $\oaresolventalgebra(\dom \mathsf{m},\sigma)$, $\alpha_{\kappa,\Lambda}$ is an automorphism group on it.
With an appropriate restriction of the domain, $\alpha$ is an automorphism group on the restricted subalgebra.
\end{prop}

\begin{proof}
Since the treatment itself does not change with only the restriction of the domain, we argue using the notation with cutoffs.
By Proposition \ref{expedition0011238},
the auxiliary functional $\mathsf{M}_{\kappa,\Lambda,t}$ is a cocycle.

(Automorphism property): Under the assumption that the first argument of the resolvent function does not become $0$,
it suffices to directly verify the preservation of the resolvent relations.

(Restricted automorphism property): Choose $f
\in \dom \mathsf{m}$ arbitrarily.
Then $\alpha_{\kappa,\Lambda,t}(\oaresolvent(z,f))
=
\oaresolvent(z + \imunit \mathsf{M}_{\kappa,\Lambda}(f), \napiernum^{\imunit t \omega} f)$ is well-defined.
In particular, $\napiernum^{\imunit t \omega} f$ satisfies $\abs{\napiernum^{\imunit t \omega} f}
= \abs{f}$, so integrability is preserved and $\dom \mathsf{m}$ is left invariant.
Therefore $\fnrestr{\alpha_{\kappa,\Lambda}}{\oaresolventalgebra(\dom \mathsf{m},\sigma)}$ is an automorphism group on $\oaresolventalgebra(\dom \mathsf{m},\sigma)$.

(Group property): It suffices to compute directly using the cocycle property of the auxiliary functional $\mathsf{M}_{\kappa,\Lambda,t}$.
\end{proof}

\begin{rem}
Here we defined the van Hove model via the automorphism group.
By examining the derivative or generator of the automorphism group, one can indeed confirm that the van Hove model is defined.
Details will be discussed in the subsequent detailed analysis of the van Hove model.
\end{rem}

Based on Proposition \ref{expedition0011614}, we define the state \(\oastate[\psi_{\txtvanhowe,I_{L}^{3},\sminvtemperature,\smchemicalpotential}]\) as the quasi-free state whose two-point function in the bounded system satisfies \begin{equation}
\begin{aligned}
&\fun{\oastate[\psi_{\txtvanhowe,I_{L}^{3},\sminvtemperature,\smchemicalpotential}]}
{\oaresolvent(\lambda,f)
\oaresolvent(\mu,g)}
\\ 
&=
\imunit
\int_0^{(\sgn \lambda) \infty}
\int_0^{(\sgn \mu) \infty}
\napiernum^{-\lambda s - \mu t}
\mathsf{T}_{\sminvtemperature,\smchemicalpotential,\kappa,\Lambda}(s,f;t,g)
\opdmsr{s}
\opdmsr{t},
\\ 
&\mathsf{T}_{\sminvtemperature,\smchemicalpotential,\kappa,\Lambda}(s,f;t,g)
\\
&=
\fnexp{-\frac{\opform{q}_{\txtnonzero,\smchemicalpotential}(sf+tg)}{4}
+\imunit \opreal \mathsf{m}_{\kappa,\Lambda}(sf+tg)
-\frac{\imunit}{2} st \opimag \bkt{f}{g}}.
\end{aligned}
\end{equation}

\begin{prop}\label{expedition0011634}
The state $\oastate[\psi_{\txtvanhowe,I_{L}^{3},\sminvtemperature,\smchemicalpotential}]$ is a $\sminvtemperature$-KMS state for the automorphism group $\alpha_{\txtvanhowe,I_{L}^{3},\kappa,\Lambda,\smchemicalpotential}$ of Definition \ref{expedition0011278}.
\end{prop}

\begin{proof}
The function $\mathsf{T}_{\sminvtemperature,\smchemicalpotential,\kappa,\Lambda}(s,f;t,g)$ is determined by the $\sminvtemperature$-KMS state on the Weyl algebra,
and in particular is determined by the trace.
The rest follows from the KMS property of the trace and the KMS condition on a dense $\ast$-algebra \cite{BratteliRobinson2}.
\end{proof}

\begin{prop}
Let the chemical potential be $\smchemicalpotential
= 0$.
Then the infinite volume limit $\oastate[\psi_{\txtvanhowe,\sminvtemperature,\smchemicalpotential}]$ as a state on $\oaresolventalgebra_{\txtirsingular,\sminvtemperature}
=
\oaresolventalgebra(\sphilb{D}_{\txtirsingular,\sminvtemperature})$ exists for the automorphism group $\alpha_{\kappa,\Lambda,\smchemicalpotential}$ of Definition \ref{expedition0011278}.
The infinite volume limit is also a $\sminvtemperature$-KMS state for the automorphism group of the infinite system.
\end{prop}

\begin{proof}
This follows by the same argument as Proposition \ref{expedition0011634}.
\end{proof}

We summarize the situation in which BEC can occur.

\begin{cor}\label{expedition0012086}
The $\sminvtemperature$-KMS state $\oastate[\psi_{\txtvanhowe,\sminvtemperature}]$ obtained in the infinite volume limit is a quasi-free state,
and the two-point function is
\begin{equation}
\begin{aligned}
\fun{\oastate[\psi_{\txtvanhowe,\sminvtemperature}]}
{\oaresolvent(\lambda,f)
\oaresolvent(\mu,g)}
&=
\imunit
\int_0^{(\sgn \lambda) \infty}
\int_0^{(\sgn \mu) \infty}
\napiernum^{-\lambda s - \mu t}
\mathsf{T}_{\sminvtemperature,\kappa,\Lambda}(s,f;t,g)
\opdmsr{s}
\opdmsr{t},
\\ 
\mathsf{T}_{\sminvtemperature,\kappa,\Lambda}(s,f;t,g)
&=
\fnexp{\imunit \opreal \mathsf{m}_{\kappa,\Lambda}(sf+tg)
-\frac{\imunit}{2} st \opimag \bkt{f}{g}}
\\
&\qquad\times
\fnexp{-\frac{\opform{q}_{\txtnonzero}(sf+tg)}{4}
-\frac{\opform{q}_0(sf+tg)}{4}}.
\end{aligned}
\end{equation}

\end{cor}

\begin{proof}
This follows from the arguments for the Weyl algebra,
in particular by Theorem \ref{expedition0011637}.
\end{proof}

\begin{rem}
The integrand $\mathsf{T}_{\sminvtemperature,\kappa,\Lambda}$ is essentially the expectation value on the Weyl algebra.
\end{rem}

As with the Weyl algebra, we remove the infrared and ultraviolet cutoffs.

\begin{thm}
\begin{enumerate}
\item
Convergence of the automorphism group: For any $A
\in \oaresolventalgebra_{\txtirsingular,\sminvtemperature}$ and $t
\in \fldreal$, the limit $$\alpha_{\txtvanhowe,t}(A)
=
\lim_{\kappa \to 0, \Lambda \to \infty}
\alpha_{\txtvanhowe,t,\kappa,\Lambda}(A)$$exists in the strong topology of $\oaresolventalgebra_{\txtirsingular,\sminvtemperature}$,
and the family $\fml{\alpha_{\txtvanhowe,t}}{t \in \fldreal}$ gives a strongly continuous one-parameter automorphism group on $\oaresolventalgebra_{\txtirsingular,\sminvtemperature}$.

\item
Convergence of the KMS state: For any $A
\in \oaresolventalgebra_{\txtirsingular,\sminvtemperature}$, there exists a state $\oastate[\psi_{\txtvanhowe,\sminvtemperature,\kappa,\Lambda}]$ on $\oaresolventalgebra_{\txtirsingular,\sminvtemperature}$ satisfying $$\oastate[\psi_{\txtvanhowe,\sminvtemperature}](A)
=
\lim_{\kappa \to 0,\Lambda \to \infty}
\oastate[\psi_{\txtvanhowe,\sminvtemperature,\kappa,\Lambda}](A).$$
In particular, for the generators of the resolvent algebra, $$\oastate[\psi_{\txtvanhowe,\sminvtemperature}]
(\oaresolvent(\lambda,f))
=
\lim_{\kappa \to 0,\Lambda \to \infty}
\oastate[\psi_{\txtvanhowe,\sminvtemperature,\kappa,\Lambda}]
(\oaresolvent(\lambda,f))
\quad
f \in \sphilb{D}_{\txtirsingular,\sminvtemperature}$$holds,
and the right-hand side is obtained as the natural limit of the expression in Theorem \ref{expedition0011612}.

\item
The state $\oastate[\psi_{\txtvanhowe,\sminvtemperature}]$ is a $\sminvtemperature$-KMS state for the automorphism group $\alpha_{\txtvanhowe}$.
\end{enumerate}
\end{thm}

\begin{proof}
This can be shown in the same way as Theorem \ref{expedition0011637}.
\end{proof}

\subsection{No-go theorem for BEC and ideal structure}\label{no-go-theorem-for-bec-and-ideal-structure}

We reformulate the results on the Weyl algebra obtained in the Subsections \ref{expedition0012081} in terms of the ideal structure of the resolvent algebra. The degrees of freedom carrying the non-condensed and condensed components, together with the constraints on physical quantities imposed by infrared singularities, are described in a unified manner as algebraic structures.

Let \(X = \opformdomain(\opform{q}_{\txtnonzero})\) denote the domain of the quasi-bilinear form \(\opform{q}_{\txtnonzero}\) corresponding to the non-condensed component, which we take as the fundamental symplectic space, and let \(\oaresolventalgebra = \oaresolventalgebra(X, \sigma)\) denote the associated resolvent algebra. This algebra contains the full degrees of freedom at the stage where neither infrared singularities nor the presence of condensation are distinguished. The condensed component is given by the quasi-bilinear form \(\opform{q}_0\), and its domain \(X_0 = \opformdomain(\opform{q}_0)\) is the subspace carrying the degrees of freedom corresponding to BEC.

Due to the presence of infrared divergences, not every test function is physically admissible. The infrared singularity determines the domain \(\dom \mathsf{m} \subset X\) of a linear functional \(\mathsf{m}\), on which physical quantities are finitely defined. Accordingly, the space of physically admissible degrees of freedom is given by \(X_{\txtphys} = X \cap \dom \mathsf{m}\).

This restriction is realized by a closed two-sided ideal of the resolvent algebra. Specifically, we define the ideal corresponding to the infrared singular directions as \[\oaideal{J}_{\txtircapitalized} = \clos{\opideal \set{\oaresolvent(\lambda, f)}{\lambda \in
\fldreal^{\txtexcludezero}, f \in X \setminus X{\txtphys}}} \subset \oaresolventalgebra.\] The resolvent algebra of physical quantities is then given by \[\oaresolventalgebra_{\txtphys} = \setquot{\oaresolventalgebra}{\oaideal{J}_{\txtircapitalized}},\] and by the universality of the resolvent algebra, we have \(\oaresolventalgebra_{\txtphys}
\eqisom \oaresolventalgebra(X_{\txtphys}, \sigma)\). Similarly, the degrees of freedom corresponding to the condensed component are described by the closed two-sided ideal \[\oaideal{J}_0 = \clos{\opideal \set{\oaresolvent(\lambda, f)}{\lambda \in \fldreal^{\txtexcludezero}, f
\in X_0}} \subset \oaresolventalgebra.\]

With these preparations, the following theorem holds.

\begin{thm}
For any $s > 2$, let the dispersion relation be $\omega_s(k) = \abs{k}^s$. Then for every $f \in
X_{\txtphys}$, we have $\opform{q}_0(f) = 0$. Consequently, $X_0 \cap X_{\txtphys} = \setone{0}$, and the
condensed component vanishes on $X_{\txtphys}$. In particular, $\oaideal{J}_0 \subset
\oaideal{J}_{\txtircapitalized}$ holds.
\end{thm}

This theorem shows that the degrees of freedom corresponding to BEC are entirely contained within the ideal determined by the infrared singularity. In particular, the resolvent algebra of physical quantities is \(\oaresolventalgebra_{\txtphys} =
\setquot{\oaresolventalgebra}{\oaideal{J}_{\txtircapitalized}}\), and on this quotient algebra \(\oaideal{J}_0 = 0\) holds.

Thus, within the framework of ideal theory for the resolvent algebra, the absence of Bose--Einstein condensation in the equilibrium state of quasi-particles is described as the elimination of the degrees of freedom carrying the condensed component by the ideal corresponding to the infrared singularity.

\bibliography{myref.bib}

\end{document}

%% file: mycommands.tex
\makeatletter
\providecommand*{\dashv}{\mathrel{\mathpalette\@Dashv\vDash}}
\newcommand*{\@dashv}[2]{\reflectbox{$\m@th#1#2$}}
\makeatother
\newcommand{\abs}[1]{\left| #1 \right|} 
\NewDocumentCommand{\weaknorm}{O{\dbk} m}{#1{#2}} 
\newcommand{\norm}[1]{\left\Vert #1 \right\Vert} 

\newcommand{\bkt}[2]{\left\langle #1,\,#2 \right\rangle} 
\newcommand{\rbkt}[2]{\left( #1,\,#2 \right)} 
\newcommand{\isomto}{\mathrel{\rightarrowtail\kern-1.9ex\twoheadrightarrow}} 
\newcommand{\cbk}[1]{\left\{ #1 \right\}} 
\newcommand{\dbk}[1]{\left\langle #1 \right\rangle} 
\newcommand{\pairbk}[1]{\rbk{#1}} 
\newcommand{\rbkleft}[1]{\left( #1 \right.} 
\newcommand{\rbkright}[1]{\left. #1 \right)} 
\newcommand{\rbk}[1]{\left( #1 \right)} 
\newcommand{\sqbk}[1]{\left[ #1 \right]} 
\newcommand{\funrbk}[2]{\fun{\rbk{#1}}{#2}} 
\newcommand{\fun}[2]{#1 \rbk{#2}} 
\newcommand{\sqfun}[2]{#1 \sqbk{#2}} 
\newcommand{\closedinterval}[2]{\sqbk{#1,\,#2}} 
\newcommand{\commutator}[2]{\sqbk{#1,\,#2}} 

\DeclareMathOperator{\sgn}{sgn} 
\NewDocumentCommand{\imunit}{O{\mathsf{i}}}{#1} 
\NewDocumentCommand{\placeholder}{O{\bullet}}{#1} 
\NewDocumentCommand{\trace}{O{\operatorname{Tr}}}{#1} 
\newcommand{\Ad}{\operatorname{Ad}} 
\newcommand{\Ker}{\operatorname{Ker}} 
\newcommand{\bigmiddleslash}[2]{\left. #1 \middle/ #2 \right.} 
\newcommand{\clos}[1]{\overline{#1}} 
\newcommand{\cmpconj}[1]{\overline{#1}} 
\newcommand{\diracdelta}{\delta} 
\newcommand{\dom}{\operatorname{dom}} 
\newcommand{\eqcsq}[1]{\sqbk{#1}} 
\newcommand{\eqisom}{\cong} 
\newcommand{\fml}[2]{\cbk{#1}_{#2}} 
\newcommand{\hyphen}{\hbox{-}} 
\newcommand{\idone}{1} 
\newcommand{\invrbk}[1]{\rbk{#1}^{-1}} 
\newcommand{\inv}[1]{#1^{-1}} 
\newcommand{\napiernum}{\mathsf{e}} 
\newcommand{\onehalf}{\frac{1}{2}} 
\newcommand{\oneoverfour}{\frac{1}{4}} 
\newcommand{\opideal}{\operatorname{Ideal}} 
\newcommand{\opimag}{\operatorname{Im}} 
\newcommand{\opod}[1]{\frac{d}{d #1}} 
\newcommand{\opreal}{\operatorname{Re}} 
\newcommand{\setSymbolDownLeft}[2]{{\vphantom{#2}}_{#1}{#2}} 
\newcommand{\setSymbolUpLeft}[2]{{\vphantom{#2}}^{#1}{#2}} 
\DeclareMathOperator{\Rep}{Rep} 
\NewDocumentCommand{\agvariety}{O{\mathcal}}{#1} 
\NewDocumentCommand{\cmdrel}{O{\omega}}{#1} 
\NewDocumentCommand{\dfsp}{O{A}}{#1} 
\NewDocumentCommand{\eqcpointed}{O{\eqcsq} m}{#1{#2}_{\ast}} 
\NewDocumentCommand{\fnheaviside}{O{H}}{#1} 
\NewDocumentCommand{\fthol}{O{\mathcal{O}}}{#1} 
\NewDocumentCommand{\ftmero}{O{\mathcal{M}}}{#1} 
\NewDocumentCommand{\grcentralizer}{O{Z}}{#1} 
\NewDocumentCommand{\grmetform}{O{2} m m}{\grmet[#1] \! \rbkt{#2}{#3}} 
\NewDocumentCommand{\grmet}{O{2}}{\setSymbolDownLeft{#1}{g}} 
\NewDocumentCommand{\grnormalizer}{O{N}}{#1} 
\NewDocumentCommand{\gropasym}{O{A}}{#1} 
\NewDocumentCommand{\gropsym}{O{S}}{#1} 
\NewDocumentCommand{\grpermorderedpair}{O{\mathcal{P}}}{#1} 
\NewDocumentCommand{\grsym}{O{\mathfrak{S}} m}{#1_{#2}} 
\NewDocumentCommand{\gtbase}{O{\mathcal}}{#1} 
\NewDocumentCommand{\gtfilter}{O{\mathcal}}{#1} 
\NewDocumentCommand{\gtfmlclosed}{O{\mathcal}}{#1} 
\NewDocumentCommand{\gtfmlopen}{O{\mathcal}}{#1} 
\NewDocumentCommand{\gtopenball}{O{U}}{#1} 
\NewDocumentCommand{\gtopencover}{O{\mathcal}}{#1} 
\NewDocumentCommand{\gtopennbh}{O{\mathcal}}{#1} 
\NewDocumentCommand{\gtpreopencover}{O{\mathcal}}{#1} 
\NewDocumentCommand{\gtsubbase}{O{\mathcal}}{#1} 
\NewDocumentCommand{\gtvicinity}{O{\mathcal}}{#1} 
\NewDocumentCommand{\lasp}{O{\mathcal}}{#1} 
\NewDocumentCommand{\latpright}{O{\top} m}{#2^{#1}} 
\NewDocumentCommand{\latp}{O{t} m}{\setSymbolUpLeft{#1}{#2}} 
\NewDocumentCommand{\lpdistribution}{O{\mu} m}{#2_{\ast,#1}} 
\NewDocumentCommand{\lpmollifier}{O{\rho}}{#1} 
\NewDocumentCommand{\lpofpositive}{O{\chi}}{#1} 
\NewDocumentCommand{\manliederiv}{O{L}}{#1} 
\NewDocumentCommand{\mansmoothnbh}{O{\mathcal}}{#1} 
\NewDocumentCommand{\mblfmldsysgenerated}{O{d} m}{\fun{#1}{#2}} 
\NewDocumentCommand{\mblfmlgenerated}{O{\sigma} m}{\fun{#1}{#2}} 
\NewDocumentCommand{\oacorrfn}{O{\Gamma}}{#1} 
\NewDocumentCommand{\oagnsvector}{O{\Omega}}{#1} 
\NewDocumentCommand{\oaideal}{O{\mathcal}}{#1} 
\NewDocumentCommand{\oanumberoperator}{O{A}}{#1} 
\NewDocumentCommand{\oaposcone}{O{\mathcal{P}}}{#1} 
\NewDocumentCommand{\oapressure}{O{P}}{#1} 
\NewDocumentCommand{\oarepn}{O{\pi}}{#1} 
\NewDocumentCommand{\oaspnormalstate}{O{N}}{#1} 
\NewDocumentCommand{\oasppurestate}{O{P}}{#1} 
\NewDocumentCommand{\oaspstate}{O{E}}{#1} 
\NewDocumentCommand{\oastatevector}{O{\Omega}}{#1} 
\NewDocumentCommand{\oastate}{O{\omega}}{#1} 
\NewDocumentCommand{\opdilation}{O{\delta}}{#1} 
\NewDocumentCommand{\opdmat}{O{\rho}}{#1} 
\NewDocumentCommand{\opfockan}{O{a}}{#1} 
\NewDocumentCommand{\opfockcran}{O{a}}{#1^{\#}} 
\NewDocumentCommand{\opfockcrdagger}{O{a}}{#1^{\dagger}} 
\NewDocumentCommand{\opfockcr}{O{a}}{#1^{\ast}} 
\NewDocumentCommand{\opfocknumber}{O{N}}{#1} 
\NewDocumentCommand{\opfocksegalconj}{O{\pi}}{#1} 
\NewDocumentCommand{\opfocksegal}{O{\phi}}{#1} 
\NewDocumentCommand{\opspecmeas}{O{E}}{#1} 
\NewDocumentCommand{\opspec}{O{} m}{\fun{\sigma_{#1}}{#2}} 
\NewDocumentCommand{\optransl}{O{\tau}}{#1} 
\NewDocumentCommand{\physaction}{O{\mathcal{A}}}{#1} 
\NewDocumentCommand{\physcharge}{O{e}}{\mathrm{#1}} 
\NewDocumentCommand{\physcplconst}{O{\mathsf{g}}}{#1} 
\NewDocumentCommand{\physelectrostaticcapasity}{O{\mathrm{Cap}}}{#1} 
\NewDocumentCommand{\physenergy}{O{E}}{#1} 
\NewDocumentCommand{\physgse}{O{E}}{#1_{0}} 
\NewDocumentCommand{\physham}{O{H}}{#1} 
\NewDocumentCommand{\physlagdensity}{O{\mathcal{L}}}{#1} 
\NewDocumentCommand{\physlag}{O{L}}{#1} 
\NewDocumentCommand{\physliouvilean}{O{L}}{#1} 
\NewDocumentCommand{\physmass}{O{m}}{#1} 
\NewDocumentCommand{\prbcharfun}{O{\chi}}{#1} 
\NewDocumentCommand{\prbdist}{O{\mathcal{P}}}{#1} 
\NewDocumentCommand{\prbgaussianmeasure}{O{\msrcal{N}}}{#1} 
\NewDocumentCommand{\prbnormaldist}{O{N}}{#1} 
\NewDocumentCommand{\prbprocess}{O{X}}{#1} 
\NewDocumentCommand{\prbqspace}{O{\mathcal{Q}}}{#1} 
\NewDocumentCommand{\prbspsample}{O{\Omega}}{#1} 
\NewDocumentCommand{\psh}{O{\mathfrak}}{#1} 
\NewDocumentCommand{\qtquantumchannel}{O{\mathcal{L}}}{#1} 
\NewDocumentCommand{\repn}{O{\pi}}{#1} 
\NewDocumentCommand{\schattencls}{O{\mathbb{K}}}{#1} 
\NewDocumentCommand{\setfmlcylinder}{O{\mathcal{C}}}{#1} 
\NewDocumentCommand{\setfml}{O{\mathcal}}{#1} 
\NewDocumentCommand{\setindex}{O{\mathcal} m}{#1{#2}} 
\NewDocumentCommand{\setlattice}{O{\Gamma}}{#1} 
\NewDocumentCommand{\setspecial}{O{\mathcal} m}{#1{#2}} 
\NewDocumentCommand{\shdiffform}{O{\sheaf{A}}}{#1} 
\NewDocumentCommand{\sheaf}{O{\mathfrak}}{#1} 
\NewDocumentCommand{\smchemicalpotential}{O{\mu}}{#1} 
\NewDocumentCommand{\smenergydensity}{O{\varrho}}{#1} 
\NewDocumentCommand{\smfluctuationwithdmat}{O{\beta} m}{\smuncertaintywithdmat[#1]{#2}^2} 
\NewDocumentCommand{\sminvtemperature}{O{\beta}}{#1} 
\NewDocumentCommand{\smlocaldensityoperator}{O{\rho}}{#1} 
\NewDocumentCommand{\smmicrocanonicalstate}{O{\beta} m}{\physmean{#2}_{#1}} 
\NewDocumentCommand{\smnumberdensity}{O{\rho}}{#1} 
\NewDocumentCommand{\smooth}{O{\mathcal{E}}}{#1} 
\NewDocumentCommand{\smparticlenumber}{O{N}}{#1} 
\NewDocumentCommand{\smpressure}{O{p}}{#1} 
\NewDocumentCommand{\smspecificfreeenergy}{O{\bar{f}}}{#1} 
\NewDocumentCommand{\smthermalvac}{O{\beta}}{\Omega_{#1}} 
\NewDocumentCommand{\smuncertaintywithdmat}{O{\beta} m}{\rbk{\triangle #2}_{#1}} 
\NewDocumentCommand{\sphilb}{O{\mathcal}}{#1} 
\NewDocumentCommand{\splowerhalf}{O{\mathbb{H}}}{#1_{\txtneg}} 
\NewDocumentCommand{\spupperhalf}{O{\mathbb{H}}}{#1_{\txtnonneg}} 
\NewDocumentCommand{\topmetric}{O{d}}{#1} 
\NewDocumentCommand{\vaoutnormal}{O{\widehat}}{#1} 
\newcommand{\category}[1]{\mathop{\mathsf{#1}}} 

\newcommand{\catpresheaf}[1]{\category{PSh}} 
\newcommand{\dstrapiddec}{\mathcal{S}} 
\newcommand{\dsttempered}{\dstrapiddec^{\ast}} 
\newcommand{\faadj}[1]{#1^{\ast}} 
\newcommand{\faftr}[1]{\widehat{#1}} 
\newcommand{\fldcmp}{\fld{C}} 
\newcommand{\fldmultiplicativegroup}[1]{#1^{\times}} 
\newcommand{\fldreal}{\fld{R}} 
\newcommand{\fld}[1]{\mathbb{#1}} 
\newcommand{\fndef}[1]{\boldsymbol{1}_{#1}} 
\newcommand{\fnexp}[1]{\fun{\exp}{#1}} 
\newcommand{\fnrestr}[2]{\left. #1 \right|_{#2}} 
\newcommand{\greuctr}[2]{\fldreal_{#1}^{#2}} 
\newcommand{\gtclos}[1]{\overline{#1}} 
\newcommand{\lp}{L} 
\newcommand{\mblfn}{M} 
\newcommand{\msrcal}[1]{\mathcal{#1}} 

\newcommand{\oacstar}{C^{\ast}} 
\newcommand{\oaresolventalgebra}{\oa{R}} 
\newcommand{\oaresolvent}{R} 
\newcommand{\oastaralgebra}{\operatorname{\ast \hyphen \mathrm{alg}}} 
\newcommand{\oaweyl}{\oa{W}} 
\newcommand{\oa}[1]{\mathcal{#1}} 
\newcommand{\opdmsr}[1]{\mathop{d #1}} 
\newcommand{\opfnresolvent}[1]{\invrbk{#1}} 
\newcommand{\opfocksndqntdiff}{d \Gamma} 
\newcommand{\opfockweyl}{W} 
\newcommand{\opformdomain}{\mathop{Q}} 
\newcommand{\opform}[1]{\mathsf{#1}} 
\newcommand{\opspbddlin}[1]{\fun{\mathbb{B}}{#1}} 
\newcommand{\opspecint}[1]{\mathcal{E}} 
\newcommand{\physmean}[1]{\dbk{#1}} 
\newcommand{\ringratint}{\mathbb{Z}} 

\newcommand{\seq}[2]{\if\relax\detokenize{#1}\relax \rbk{#1} \else \rbk{#1}_{#2} \fi} 
\newcommand{\setcardop}[1]{\operatorname{card} #1} 
\newcommand{\setisomorphism}[1]{\operatorname{Iso}} 
\newcommand{\setone}[1]{\cbk{#1}}
\newcommand{\setquot}[2]{\bigmiddleslash{#1}{#2}} 
\newcommand{\set}[2]{\left\{#1 \, \middle| \, #2\right\}}
\newcommand{\smgrandpartitionfunc}{\Xi} 
\newcommand{\spfock}{\mathcal{F}} 
\newcommand{\txtbec}{\mathrm{BEC}} 
\newcommand{\txtborel}{\mathrm{b}} 
\newcommand{\txtbsn}{\mathrm{b}} 
\newcommand{\txtdiscrete}{\mathrm{d}} 
\newcommand{\txtexcludezero}{\times} 
\newcommand{\txtfin}{\mathrm{fin}} 
\newcommand{\txtfock}{\mathrm{F}} 
\newcommand{\txtfr}{\mathrm{fr}} 
\newcommand{\txtircapitalized}{\mathrm{IR}} 
\newcommand{\txtirsingular}{\mathrm{irs}} 
\newcommand{\txtneg}{\mathrm{-}} 
\newcommand{\txtnonneg}{\mathrm{+}} 
\newcommand{\txtnonzero}{\mathrm{nz}} 
\newcommand{\txtphys}{\mathrm{phys}} 
\newcommand{\txtreal}{\mathrm{real}} 
\newcommand{\txtgrandcanonical}{\mathrm{GC}} 
\newcommand{\txtselfconsistent}{\mathrm{sc}} 
\newcommand{\txtsym}{\mathrm{s}} 
\newcommand{\txtvanhowe}{\mathrm{vH}} 